\documentclass[runningheads,envcountsame,envcountsect,orivec]{llncs}
\usepackage[T1]{fontenc}
\usepackage{graphicx}
\usepackage{color}
\usepackage[hidelinks]{hyperref}

\usepackage{bbding} 
\usepackage{thm-restate}
\usepackage{amsfonts}
\usepackage{amssymb}
\usepackage{latexsym}
\usepackage{scalerel}
\usepackage{proof}
\usepackage{ifthen}
\usepackage{mathtools}
\usepackage{cleveref}

\usepackage{tikz}
\usetikzlibrary{positioning,arrows,calc}

\renewcommand{\qed}{\hfill$\Box$}

\newcommand{\m}[1]{\mathsf{#1}}
\newcommand{\mr}[1]{\mathrel{#1}}
\newcommand{\seq}[2][n]{{#2_1},\dots,{#2_{#1}}}
\newcommand{\sig}[2][n]{{#2_1}\times\cdots\times{#2_{#1}}}
\newcommand{\h}[1][.3]{\hspace{#1mm}}
\newcommand{\SET}[1]{\{\h#1\h\}}
\newcommand{\inter}[1]{[\![{#1}]\!]}

\newcommand{\xB}{\mathcal{B}}

\newcommand{\xE}{\mathcal{E}}
\newcommand{\xF}{\mathcal{F}}

\newcommand{\xI}{\mathcal{I}}
\newcommand{\xJ}{\mathcal{J}}
\newcommand{\xM}{\mathcal{M}}
\newcommand{\xR}{\mathcal{R}}
\newcommand{\xS}{\mathcal{S}}
\newcommand{\xT}{\mathcal{T}}

\newcommand{\xV}{\mathcal{V}}
\newcommand{\xFTe}{\xF_{\m{te}}}
\newcommand{\xFTh}{\xF_{\m{th}}}
\newcommand{\xSTe}{\xS_{\m{te}}}
\newcommand{\xSTh}{\xS_{\m{th}}}

\newcommand{\Val}{\xV\m{al}}
\newcommand{\Var}{\xV\m{ar}}
\newcommand{\FVar}{\xF\xV\m{ar}}
\newcommand{\BVar}{\xB\xV\m{ar}}

\newcommand{\LVar}{\mathcal{L}\Var}

\newcommand{\ExVar}{\mathcal{E}\m{x}\Var}
\newcommand{\Dom}{\mathcal{D}\m{om}}

\newcommand{\VDom}{\mathcal{V}\Dom}
\newcommand{\Pos}{\mathcal{P}\m{os}}

\newcommand{\FPos}{\Pos_\xF}

\newcommand{\sort}[1]{\m{#1}}
\newcommand{\Bool}{\sort{Bool}}

\newcommand{\CO}[1]{[\h#1\h]} 
\newcommand{\ECO}[2]{\exists #1.\ #2}

\newcommand{\CTerm}[4]{%
\ifthenelse{\equal{#1}{} \and \equal{#2}{} \and \equal{#3}{} \and \equal{#4}{}}%
{\mathrm{\Pi} X.\ s~\CO{\ECO{\vec{x}}{\varphi}}}%
{%
\ifthenelse{\equal{#1}{}}{}{\mathrm{\Pi} #1.\ }%
#2%
\ifthenelse{\equal{#4}{}}%
{}%
{\ifthenelse{\equal{#3}{}}{~\CO{#4}}{~\CO{\ECO{#3}{#4}}}}%
}%
}

\newcommand{\CEqn}[4]{%
\ifthenelse{\equal{#1}{}}{%
\ifthenelse{\equal{#4}{}}{#2 \approx #3}{#2 \approx #3~\CO{#4}}}{%
\ifthenelse{\equal{#4}{}}{\mathrm{\Pi} #1.\, #2 \approx #3}{\mathrm{\Pi} #1.\, #2 \approx #3~\CO{#4}}%
}}

\newcommand{\CRu}[4]{%
\ifthenelse{\equal{#1}{} \and \equal{#2}{} \and \equal{#3}{} \and \equal{#4}{}}%
{\mathrm{\Pi} X.\ \ell \R r~\CO{\varphi}}%
{%
\ifthenelse{\equal{#1}{}}{}{\mathrm{\Pi} #1.\ }%
#2 \R #3%
\ifthenelse{\equal{#4}{}}%
{}%
{~\CO{#4}%
}%
}%
}

\newcommand{\R}{\rightarrow}

\newcommand{\Rs}{\stackrel{\smash{\raisebox{-.5mm}{\tiny $\sim$~}}}{\R}}

\renewcommand{\L}{\leftarrow}

\newcommand{\Lb}[1][]{\mr{\vphantom{\R}_{#1}{\L}}}

\newcommand{\Rbase}[1][]{\R_{\mathsf{base}}}

\newcommand{\Ca}[1][]{\xleftrightarrow{#1}}

\newcommand{\Cb}[1][]{\Ca[]_{#1}}

\newcommand{\Cru}[1][\xE]{\ifthenelse{\equal{#1}{}}{\Cb[\m{rule}]}{\Cb[\m{rule},#1]}}
\newcommand{\Cbase}[1][\xE]{\ifthenelse{\equal{#1}{}}{\Cb[\m{base}]}{\Cb[\m{base},#1]}}

\renewcommand{\geq}{\geqslant}
\renewcommand{\leq}{\leqslant}
\renewcommand{\ge}{\geqslant}
\renewcommand{\le}{\leqslant}

\newcommand\subsetsim{\mathrel{\substack{
  \textstyle\subset\\[-0.2ex]\textstyle\sim}}}
\newcommand\supsetsim{\mathrel{\substack{
  \textstyle\supset\\[-0.2ex]\textstyle\sim}}}

\newcommand{\PG}{\mathrm{PG}} 
\newcommand{\lvf}{\mathrm{lvf}}
\newcommand{\ext}{\mathrm{ext}}
\newcommand{\rmv}{\mathrm{rmv}}

\newcommand{\Bfnum}[1]{(#1)}

\newcommand{\pvec}[1]{\vec{#1}\mkern2mu\vphantom{#1}}

\begin{document}
\title{%
Characterizing Equivalence of Logically Constrained Terms via Existentially Constrained Terms%
\thanks{This work was partially supported by 
JSPS KAKENHI Grant Numbers JP24K14817 and JP24K02900, 
and
FWF (Austrian Science Fund) project I~5943-N.}
}
\titlerunning{Characterizing Equivalence of Logically Constrained Terms}
%
\author{%
Kanta Takahata\inst{1}
\and 
Jonas Sch\"opf\inst{2}%
\Envelope\orcidID{0000-0001-5908-8519}
\and
Naoki Nishida\inst{3}%
\Envelope\orcidID{0000-0001-8697-4970}
\and
Takahito Aoto\inst{1}%
\Envelope\orcidID{0000-0003-0027-0759}
}
\authorrunning{K. Takahata, J. Sch\"opf, N. Nishida, and T. Aoto}
%
\institute{%
Niigata University, Japan\\
\email{aoto@ie.niigata-u.ac.jp}
\and
University of Innsbruck, Austria\\
\email{jonas.schoepf@uibk.ac.at}
\and
Nagoya University, Japan\\
\email{nishida@i.nagoya-u.ac.jp}
}
\maketitle              
\begin{abstract}
Logically constrained term rewriting is a rewriting 
framework that supports built-in data structures such as 
integers and bit vectors. Recently, constrained terms 
play a key role in various analyses and applications of 
logically constrained term rewriting. A fundamental question 
on constrained terms arising there is how to characterize equivalence 
between them. However, in the current literature only limited 
progress has been made on this.
In this paper, we provide several sound and complete solutions to tackle this 
problem. Our key idea is the introduction of a novel concept, 
namely existentially constrained terms, into which the original form 
of constrained terms can be embedded. We present several 
syntactic characterizations of equivalence between 
existentially constrained terms. In particular, we provide 
two different kinds of complete characterizations: one is designed 
to facilitate equivalence checking, while the other is 
intended for theoretical analysis.
\end{abstract}
\section{Introduction}

The basic formalism of term rewriting is a purely syntactic computational model; 
due to its simplicity, it is one of the most extensively studied computational models.
One of the main issues of term rewriting and its real-world applications
is that the basic formalism lacks painless treatment of
built-in data structures, such as integers, bit vectors, etc.
Logically constrained term rewriting~\cite{KN13frocos} is a relatively new extension 
of term rewriting that intends to overcome such weaknesses of the basic formalism,
while keeping succinctness for theoretical analysis as a computational model.
Rewrite rules, that are used to model computations, in logically
constrained term rewrite systems (LCTRSs) are equipped with constraints over
some arbitrary theory, e.g., linear integer arithmetic.
Built-in data structures are represented via the satisfiability of
constraints within a respective theory. 
Recent progress on the LCTRS formalism was for example made in
confluence analysis~\cite{SM23,SMM24}, (non-)termination analysis~\cite{K16,NW18}, completion~\cite{WM18},
rewriting induction~\cite{KN14aplas,FKN17tocl}, algebraic semantics~\cite{ANS24}, and complexity analysis~\cite{WM21}.

Recently, constrained terms play a key role in various analyses and applications of 
logically constrained term rewriting. 
Here, a \emph{constrained term} consists of a term and a constraint,
which restricts the possibilities in which the term is instantiated.
For example, $\m{f}(x)~\CO{x > 2}$ is a constrained term (in LCTRS notation)
which can be intuitively considered as 
a set of terms $\SET{ \m{f}(x) \mid x > 2 }$.
Two constrained terms are said to be \emph{equivalent}
if they denote the same sets of terms.
A fundamental question on constrained terms arising there 
is how to characterize equivalence between them. 
However, in the current literature,
only limited progress has been made on this.

In this paper, we provide several sound and complete solutions to tackle this 
problem. Our key idea is the introduction of a novel concept, 
namely existentially constrained terms, into which the original form 
of constrained terms can be embedded. 
The idea of existentially constrained terms is very simple---just distinguish
variables that appear solely in the constraint but not in the term itself by using
existential quantifiers. 
Nevertheless, the introduction of existential quantifiers takes us a step further 
in achieving a \emph{syntactic analysis of the equivalence of constrained terms}.

During the analysis of LCTRSs, rewriting constrained terms, called
\emph{constrained rewriting}, is frequently used. It is the key part for many
different analysis techniques, e.g., for finding specific joining sequences in
confluence analysis or is part of several inference rules in rewriting
induction. Unfortunately, it is infeasible for LCTRS tools to fully support
constrained rewriting due to the heavy non-determinism mainly caused by the
equivalence transformations. To overcome this problem, \emph{most general
constrained rewriting} has been proposed and commutation of the constrained
rewriting and the equivalence transformation has been shown for left-linear
LCTRSs~\cite{TSNA25PPDP}. The characterizations on equivalence of existentially
constrained terms in this paper provide the theoretical foundation for this work.

The remainder of the paper is organized as follows.
After explaining some basics of logically constrained term rewriting
in \Cref{sec:preliminaries},
we introduce existentially constrained terms and their equivalence
in \Cref{ssec:existentially cterms and equivalence}.
Then, in \Cref{sec:Equivalence of Variable Renamed Constrained Terms},
we consider equivalence of existentially constrained terms consisting of renamed variants;
we additionally give a counterexample to a partial characterization known in
the literature and its correction.
In \Cref{ssec:pattern general constrained terms},
we introduce 
\emph{pattern-general} existentially constrained terms
and
show that any 
existentially constrained term can be transformed into
an equivalent pattern-general one.
Based on that, we give a sound and complete characterization
of equivalent pattern-general existentially constrained terms.
Finally, in \Cref{sec:general-characterization-of-equivalence}
we extend our characterization to the equivalence of 
existentially constrained terms in general.
Due to the page limit, omitted (or detailed) proofs are presented in  \Cref{sec:missing-proofs}.

\section{Preliminaries}
\label{sec:preliminaries}

In this section, we briefly recall the basic notions of LCTRSs~\cite{KN13frocos,SM23,SMM24,ANS24}
and fix additional notations used throughout this paper.
Familiarity with the basic notions of term rewriting is assumed (e.g.\ see~\cite{BN98,O02}).

An LCTRS is based on a sorted signature, thus our signature consists of
a set $\xS$ of sorts and a set $\xF$ of function symbols, where each 
$f \in \xF$ is attached with its sort declaration as $\mathrm{sort}(f) = \sig{\tau} \to \tau_0$;
we usually write $f\colon \sig{\tau} \to \tau_0$ for simplicity.
We assume that these sets can be partitioned into
two disjoint sets, i.e., $\xS = \xSTh \uplus \xSTe$ and $\xF = \xFTh \uplus \xFTe$,
where each $f\colon \sig{\tau} \to \tau_0 \in \xFTh$ satisfies 
$\tau_i \in \xSTh$ for all $0 \leqslant i \leqslant n$. 
Elements of $\xSTh$ ($\xFTh$) and $\xSTe$ ($\xFTe$)
are called theory sorts (symbols) and term sorts (symbols).
Each variable and term is equipped with a sort,
where 
the sort of $f(t_1,\ldots,t_n)$ with $\mathrm{sort}(f) = \sig{\tau} \to \tau_0$ is $\tau_0$ and
a term $t$ of sort $\tau \in \xS$ is denoted by $t^\tau$.
The sets of $\xS$-sorted variables and terms are denoted by $\xV$ and
$\xT(\xF,\xV)$. For any $T \subseteq \xT(\xF,\xV)$,
the set $T^\tau$ consists of terms in $T$ whose sort is $\tau$.
We assume a special sort $\Bool \in \xSTh$, and 
call the terms in $\xT(\xFTh,\xV)^\Bool$ \emph{logical constraints}.
Note that every variable in a logical constraint has a theory sort.
We denote the set of variables appearing in terms $\seq{t}$ by $\xV(\seq{t})$.
Sometimes sequences of variables and terms are written as $\vec{x}$ and $\vec{t}$.
The set of variables occurring in $\vec{x}$ is denoted by $\SET{\vec{x}}$.
The set of sequences of elements of a set $T$ by $T^*$,
such that $\vec{x} \in \xV^*$.

The set of positions in a term $t$ is denoted by $\Pos(t)$.
The symbol and subterm occurring at a position $p \in \Pos(t)$ is denoted
by $t(p)$ and $t|_p$, respectively. For $U \subseteq \xF \cup \xV$,
we write $\Pos_U(t) = \SET{p \in \Pos(t) \mid t(p) \in U}$
for positions with symbols in $U$.
A term obtained from $t$ by replacing subterms at parallel positions $\seq{p}$ 
by the terms $\seq{t}$, having the same sort as $t|_{p_1},\ldots,t|_{p_n}$,
is written as $t[\seq{t}]_{\seq{p}}$ or just $t[\seq{t}]$ when no confusions arises.
Sometimes we consider an expression obtained by
replacing those subterms $t|_{p_1},\ldots,t|_{p_n}$ in $t$ by holes of the same sorts,
which is called a multihole context and denoted by $t[~]_{\seq{p}}$.

A sort-preserving function $\sigma$ from $\xV$ to $\xT(\xF,\xV)$ is called
a substitution, where it is identified with its homomorphic extension 
$\sigma\colon \xT(\xF,\xV) \to \xT(\xF,\xV)$.
For a set of terms $T$, we write $\sigma(T) = \SET{\sigma(t) \mid t \in T}$;
the domain of $\sigma$ is denoted by $\Dom(\sigma)$: $\Dom(\sigma)=\SET{x \in \xV \mid \sigma(x) \neq x}$.
A substitution $\sigma$ is written as $\sigma\colon U \to T$ if 
$U \supseteq \Dom(\sigma)$ and $\sigma(U) \subseteq T$.
For a set $U \subseteq \xV$, a substitution $\sigma|_U$ is given
by $\sigma|_U(x) = \sigma(x)$ if $x \in U$ and $\sigma|_U(x) = x$ otherwise.
For substitutions $\sigma_1,\sigma_2$ such that
$\Dom(\sigma_1) \cap \Dom(\sigma_2) = \varnothing$,
the substitution $\sigma_1 \cup \sigma_2$
is given by $(\sigma_1 \cup \sigma_2)(x) = \sigma_i(x)$ if $x \in \Dom(\sigma_i)$
and $(\sigma_1 \cup \sigma_2)(x) = x$ otherwise.
A substitution $\sigma\colon \SET{\seq{x}} \to \SET{\seq{t}}$
such that $\sigma(x_i)= t_i$ is denoted by $\SET{x_1 \mapsto t_1,\ldots,x_n \mapsto t_n}$;
for brevity sometimes we write just $\left\{\vec{x} \mapsto \vec{t}\; \right\}$.
A bijective substitution $\sigma\colon \xV \to \xV$ is called a renaming,
and its inverse is denoted by $\sigma^{-1}$.

A model over a sorted signature $\langle \xSTh,\xFTh \rangle$
consists of the two interpretations $\xI$ for sorts
and $\xJ$ for function symbols.
$\xI$ assigns a non-empty set $\xI(\tau)$ to $\tau \in \xSTh$,
and $\xJ$ assigns a function $\xJ(f)\colon \xI(\tau_1) \times \cdots \times \xI(\tau_n) \to \xI(\tau_0)$
to $f\colon\sig{\tau} \to \tau_0 \in \xFTh$.
We assume a fixed model $\xM = \langle \xI, \xJ \rangle$ 
over the sorted signature $\langle \xSTh,\xFTh \rangle$
such that any element $a \in \xI(\tau)$ appears as a constant $a^\tau$ in 
$\xFTh$. These constants are called \emph{values} and the set of all values is denoted by $\Val$.
For a term $t$, we define $\Val(t) = \SET{ t(p) \mid p \in \Pos_{\Val}(t)}$
and for a substitution $\gamma$, we define $\VDom(\gamma) = \SET{x \in \xV \mid \gamma(x) \in \Val}$.
Throughout this paper we assume
the standard interpretation for the sort $\Bool \in \xSTh$,
namely $\xI(\Bool) = \mathbb{B} = \SET{\mathsf{true}, \mathsf{false}}$,
and the existence of necessary standard theory function symbols 
such as $\neg$, ${\land}$, ${\Rightarrow}$, ${=^\tau}$, etc. 
These are required to express constraints together with their
default sorts and interpretations.

A valuation $\rho$ on the model $\xM = \langle \xI, \xJ \rangle$
is a mapping that assigns any $x^\tau \in \xV$ with $\tau \in \xSTh$ to $\rho(x) \in \xI(\tau)$.
The interpretation of a term $t^{\tau} \in \xT(\xFTh, \xV)$ with $\tau \in \xSTh$ in the model $\xM$ over the valuation $\rho$
is denoted by $\inter{t}_{\xM,\rho}$.
For a valuation $\rho$ over the model $\xM$, we write $\vDash_{\xM,\rho} \varphi$ if $\inter{\varphi}_{\xM,\rho} = \mathsf{true}$, 
and $\vDash_{\xM} \varphi$ if 
$\vDash_{\xM,\rho} \varphi$ for all valuations $\rho$.
For $X \subseteq \xV$, a substitution $\gamma$ is said to be $X$-valued
if $\gamma(X) \subseteq \Val$.
We write $\gamma \vDash_\xM \varphi$ (and say $\gamma$ respects $\varphi$)
if the substitution $\gamma$ is $\Var(\varphi)$-valued 
and $\vDash_{\xM} \varphi\gamma$.
If no confusion arises then we drop the subscript $\xM$ in these notations.

\section{Existentially Constrained Terms}
\label{ssec:existentially cterms and equivalence}

In this section, we introduce our novel notion of existentially constrained terms,
and give the definition of their equivalence.
We start by giving the definition of existential constraints and then proceed
to existentially constrained terms in order to show their specific properties.

\begin{definition}[Existential Constraints]
\label{def:existantial constraints}
An \emph{existential constraint} is a pair $\langle \vec{x},
\varphi \rangle$ of a sequence of variables $\vec{x}$ and a logical constraint
$\varphi$, written as $\ECO{\vec{x}}{\varphi}$, 
such that $\SET{\vec{x}} \subseteq \Var(\varphi)$.
We define the
sets of \emph{free} and \emph{bound} variables of $\ECO{\vec{x}}{\varphi}$
as follows:
\Bfnum{1}
$\FVar(\ECO{\vec{x}}{\varphi}) = \Var(\varphi) \setminus \SET{\vec{x}}$, and
\Bfnum{2}
$\BVar(\ECO{\vec{x}}{\varphi}) = \SET{\vec{x}}$.
\end{definition}
We may abbreviate 
$\ECO{\vec{x}}{\varphi}$
to $\varphi$ if $\vec{x}$ is the empty sequence.

\begin{example}
\label{ex:existential constraints}
The following are existential constraints:
$\ECO{x}{(x \ge 3) \land (y = x * z)}$,
$\ECO{x,y,z}{(y = x * z)}$, $(y = x * z)$, and $\m{true}$.
However, $(\ECO{x}{y = 4 * z}) \land (\ECO{z}{y = 3 * z})$ is
not an existential constraint because it does not adhere to the form $\langle \vec{x}, \varphi \rangle$.
This is also not the case for $\ECO{y}{x \ge 3}$, because $\SET{y} \not\subseteq \SET{x} =  \Var(x \ge 3)$.
We have
$\FVar(\ECO{x}{(x \ge 3) \land (y = x * z)}) = \SET{y,z}$ and $\BVar(\ECO{x}{(x
\ge 3) \land (y = x * z)}) = \SET{x}$.
\end{example}

Remark here that, despite its name and form, $\ECO{\vec{x}}{\varphi}$ is not
defined as a usual existentially quantified formula, so that it
is not identified modulo renaming of the bound variables $\vec{x}$, ensuring
that $\BVar$ is well-defined. In fact, ``existential quantification'' is only
considered when the constraint is interpreted in the model 
(see \Cref{def:validity-satisfiability-existential-constraints}).

\begin{restatable}{lemma}{LemmaIIIii}
\label{lem:property existential constraints}
Let $s$ be a term and $\ECO{\vec{x}}{\varphi}$ an existential constraint.
Then, the following statements are equivalent:
\Bfnum{1}
    $\FVar(\ECO{\vec{x}}{\varphi}) \subseteq \Var(s)$ and $\BVar(\ECO{\vec{x}}{\varphi}) \cap \Var(s) = \varnothing$
\Bfnum{2}
    $\BVar(\ECO{\vec{x}}{\varphi}) = \Var(\varphi) \setminus \Var(s)$
\Bfnum{3}
    $\FVar(\ECO{\vec{x}}{\varphi}) = \Var(\varphi) \cap \Var(s)$
\end{restatable}

\begin{proof}[Sketch]
Show \Bfnum{1} $\implies$ \Bfnum{2}, \Bfnum{2} $\iff$ \Bfnum{3}, and
\Bfnum{3} $\implies$ \Bfnum{1}.
\qed
\end{proof}

As the name suggests, the existential quantifier in existential constraints
is interpreted on the model as in the predicate logic. Since constants in 
our model $\xM$ appear as constants (values) in our signature, we define validity and satisfiability as follows.

\begin{definition}[Validity and Satisfiablity of Existential Constraints]
\label{def:validity-satisfiability-existential-constraints}
Let $\ECO{\vec{x}}{\varphi}$ be an existential constraint,
and $\rho$ a valuation.
Then, we write $\vDash_{\xM,\rho} \ECO{\vec{x}}{\varphi}$ if 
there exists $\vec{v} \in \Val^*$ such that $\vDash_{\xM,\rho} \varphi\kappa$,
where $\kappa = \SET{\vec{x} \mapsto \vec{v}}$.
An existential constraint $\ECO{\vec{x}}{\varphi}$ is said to be \emph{valid}, 
written as $\vDash_{\xM} \ECO{\vec{x}}{\varphi}$, if 
$\vDash_{\xM,\rho} \ECO{\vec{x}}{\varphi}$ for any valuation $\rho$.
An existential constraint $\ECO{\vec{x}}{\varphi}$ is said to be \emph{satisfiable} if 
$\vDash_{\xM,\rho} \ECO{\vec{x}}{\varphi}$ for some valuation $\rho$.
For any substitution $\sigma$,
we write $\sigma \vDash_\xM \ECO{\vec{x}}{\varphi}$ 
(and say $\sigma$ respects $\ECO{\vec{x}}{\varphi}$) 
if $\sigma(\FVar(\ECO{\vec{x}}{\varphi})) \subseteq \Val$
and $\vDash_\xM (\ECO{\vec{x}}{\varphi})\sigma$.
Here, $(\ECO{\vec{x}}{\varphi})\sigma$ denotes the application of 
a substitution $\sigma$: 
$(\ECO{\vec{x}}{\varphi})\sigma 
:= \ECO{\vec{x}}{(\varphi\sigma|_{\FVar(\ECO{\vec{x}}{\varphi})})}$.

\end{definition}

\begin{example}[Cont'd from \Cref{ex:existential constraints}]
\label{ex:constraint respect substitution}
Let $\mathbb{Z}$ be the standard integer model
and $\rho$ be a valuation such that $\rho(y) = 6$ and $\rho(z) = 2$.
Then $\vDash_{\mathbb{Z},\rho} \ECO{x}{(x \ge 3) \land (y = x \times z)}$,
because by taking $\kappa = \SET{x \mapsto 3}$
we obtain $\vDash_{\mathbb{Z},\rho} (3 \ge 3) \land (y = 3 \times z)$.
Thus, $\ECO{x}{(x \ge 3) \land (y = x \times z)}$ is satisfiable.
Note that it is not valid---as witness consider a valuation $\rho'$ 
with $\rho'(y) = 5$ and $\rho'(z) = 2$.
For other examples $\ECO{x}{(x \ge 3) \land (y \le 3) \land (x < y)}$ is not satisfiable,
while $\ECO{x}{(x \ge 3) \land (y \le x)}$ is valid.
Let $\sigma = \SET{y \mapsto 3+3, z \mapsto 2}$ be a substitution.
Then $\sigma \not\vDash_\mathbb{Z} \ECO{x}{(x \ge 3) \land (y = x \times z)}$,
as $\sigma(y) \notin \Val$.
On the other hand,
we have $\sigma' \vDash_\mathbb{Z} \ECO{x}{(x \ge 3) \land (y = x \times z)}$
for a substitution $\sigma' = \SET{x \mapsto 1, y \mapsto 6, z \mapsto 2}$
as $\SET{y,z} \subseteq \VDom(\sigma')$ and 
the application of $\sigma'$ gives $(\ECO{x}{(x \ge 3) \land (6 = x \times 2)})$.
\end{example}

Now, we use existential constraints to define existentially constrained terms.

\begin{definition}[Existentially Constrained Term] 
\label{def:constrained-term}
An \emph{existentially constrained term} is a 
triple $\langle X, s, \ECO{\vec{x}}{\varphi} \rangle$, written as $\CTerm{X}{s}{\vec{x}}{\varphi}$, 
of a set $X$ of variables,
a term $s$, and an existentially constraint $\ECO{\vec{x}}{\varphi}$,
such that
\Bfnum{1}
    $\FVar(\ECO{\vec{x}}{\varphi}) \subseteq X \subseteq \Var(s)$,
        and
\Bfnum{2}
    $\BVar(\ECO{\vec{x}}{\varphi}) \cap \Var(s) = \varnothing$.
Variables in $X$ are called \emph{logical variables} (of $\CTerm{X}{s}{\vec{x}}{\varphi}$).
\end{definition}

\begin{example} 
\label{ex:existential-Constrained-Term}
Let $\mathsf{f},\mathsf{g} \in \xFTe$.
From this we construct the three existentially constrained terms
$\CTerm{\SET{x}}{\mathsf{g}(x)}{y}{x = 3 \times y}$, 
$\CTerm{\SET{x}}{\mathsf{f}(x+2,\mathsf{g}(y))}{z}{x = 2 \times z}$,
and 
$\CTerm{\varnothing}{\mathsf{g}(x)}{}{\m{true}}$.
The last one is abbreviated as $\CTerm{\varnothing}{\mathsf{g}(x)}{}{}$.
None of the following expressions are existentially constrained terms:
$\CTerm{\SET{y}}{\mathsf{g}(x)}{z}{z \ge y}$
($y$ does not appear in $\mathsf{g}(x)$),
$\CTerm{\SET{y}}{\mathsf{g}(x)}{x}{x \ge y}$
(the bound variable $x$ appears in $\mathsf{g}(x)$), and
$\CTerm{\SET{x}}{\mathsf{f}(x,y)}{z}{z \ge y}$
($y$ is not 
a member of the set $\SET{x}$).
\end{example}

Note that the conditions~\Bfnum{1},\,\Bfnum{2} of existentially constrained terms
equals the condition~\Bfnum{1} of \Cref{lem:property existential constraints}
including the condition on $X$ (i.e., $\FVar(\ECO{\vec{x}}{\varphi}) \subseteq X \subseteq \Var(s)$).
Thus, \Cref{lem:property existential constraints}
states useful equivalences w.r.t.\ the variable conditions of existentially constrained terms 
as follows:

\begin{restatable}{lemma}{LemmaIIIiv}
Let $X$ be a set of variables, $s$ a term, and
$\ECO{\vec{x}}{\varphi}$ an existential constraint. Then, the
following statements are equivalent:
\begin{enumerate}
\renewcommand{\labelenumi}{\Bfnum{\arabic{enumi}}}
    \item 
    $\CTerm{X}{s}{\vec{x}}{\varphi}$ is an existentially constrained term,
    \item
    $\FVar(\ECO{\vec{x}}{\varphi}) \subseteq X \subseteq \Var(s)$ and $\BVar(\ECO{\vec{x}}{\varphi}) = \Var(\varphi) \setminus \Var(s)$, and
    \item 
    $\FVar(\ECO{\vec{x}}{\varphi}) \subseteq X \subseteq \Var(s)$ and $\FVar(\ECO{\vec{x}}{\varphi}) = \Var(\varphi) \cap \Var(s)$.
\end{enumerate}
\end{restatable}

We give now the notion of equivalence of constrained terms together 
with the notion of subsumption.

\begin{definition}[Subsumption and Equivalence]
\label{def:subsumption and equivalence}
An existentially constrained term $\CTerm{X}{s}{\vec{x}}{\varphi}$ is said to be
\emph{subsumed by} an existentially constrained term
$\CTerm{Y}{t}{\vec{y}}{\psi}$, denoted by $\CTerm{X}{s}{\vec{x}}{\varphi}
\subsetsim \CTerm{Y}{t}{\vec{y}}{\psi}$, if for all $X$-valued substitutions
$\sigma$ with 
$\sigma \vDash_\xM \ECO{\vec{x}}{\varphi}$ 
there exists a
$Y$-valued substitution $\gamma$ with 
$\gamma \vDash_\xM \ECO{\vec{y}}{\psi}$
such that $s\sigma = t\gamma$.
Two existentially constrained terms
$\CTerm{X}{s}{\vec{x}}{\varphi}$
and $\CTerm{Y}{t}{\vec{y}}{\psi}$ are said to be \emph{equivalent},
denoted by $\CTerm{X}{s}{\vec{x}}{\varphi} \sim \CTerm{Y}{t}{\vec{y}}{\psi}$, 
if $\CTerm{X}{s}{\vec{x}}{\varphi} \subsetsim \CTerm{Y}{t}{\vec{y}}{\psi}$
and 
$\CTerm{X}{s}{\vec{x}}{\varphi} \supsetsim \CTerm{Y}{t}{\vec{y}}{\psi}$.
\end{definition}

\begin{example} 
\label{ex:subsumption and equivalence}
Consider $\mathsf{f} \in \xFTe$ for a signature of an LCTRS over the theory of integers.
We have
$\CTerm{\varnothing}{\mathsf{f}(1,1)}{}{\m{true}}
\sim
\CTerm{\SET{x,y}}{\mathsf{f}(x,y)}{}{(x = y) \land (x = 1)}
\sim
\CTerm{\SET{y}}{\mathsf{f}(1,y)}{}{y = 1}
\sim \CTerm{\SET{x}}{\mathsf{f}(x,x)}{}{x = 1}$.
Moreover, the following subsumptions hold:
$\CTerm{\SET{x}}{\mathsf{f}(x,y)}{}{x = 1}
\subsetsim \CTerm{\SET{x}}{\mathsf{f}(x,y)}{}{x \ge 1}$,
$\CTerm{\SET{x}}{\mathsf{f}(x,x)}{}{}
\subsetsim \CTerm{\SET{x,y}}{\allowbreak\mathsf{f}(x,y)}{}{}$, and
$\CTerm{\SET{x,y}}{\mathsf{f}(x,y)}{}{}
\subsetsim \CTerm{\SET{x}}{\mathsf{f}(x,y)}{}{}$.
However, equivalence does not hold for any these.
\end{example}

We define satisfiability of existentially constrained terms.

\begin{definition}[Satifiability of Existentially Constrained Terms]
An existentially constrained term $\CTerm{X}{s}{\vec{x}}{\varphi}$ is said to be \emph{satisfiable} if 
$\ECO{\vec{x}}{\varphi}$ is satisfiable.
\end{definition}

Below, we often focus on characterizations of equivalence 
for \emph{satisfiable} existentially constrained terms.
However, since all unsatisfiable existentially constrained terms are equivalent,
this is no restriction.

The following lemma on satisfiable existentially constrained terms turns out 
to be very useful in later proofs.

\begin{restatable}{lemma}{LemmaBasicPropertiesOfXvaluedSubstitution}
\label{lem:basic-properties-of-X-valued-substitutions}
Let $\CTerm{X}{s}{\vec{x}}{\varphi}$ be a satisfiable existentially constrained term.
Then, there exists an $X$-valued substitution $\gamma$ such that 
$\gamma \vDash_\xM \ECO{\vec{x}}{\varphi}$,  $\Dom(\gamma) = X$,
and $\Pos(s) = \Pos(s\gamma)$.
\end{restatable}

\begin{proof}
By assumption, 
the constraint $\ECO{\vec{x}}{\varphi}$ is satisfiable.
Thus, there exists a valuation $\rho$ such that 
$\vDash_{\xM,\rho} \ECO{\vec{x}}{\varphi}$.
Take a substitution $\gamma := \rho|_{X}$.
\qed
\end{proof}
 
Before concluding this section,
we relate the notion of existentially constrained terms
to the respective original notion introduced in~\cite{KN13frocos}. 
A term $\varphi$ is a logical constraint if $\varphi \in \xT(\xFTh,\xV)$ 
with a boolean sort. A logical constraint $\varphi$ is respected by a
substitution $\sigma$, written also as $\sigma \vDash \varphi$, 
if $\sigma(\Var(\varphi)) \subseteq \Val$ and $\vDash_\xM \varphi\sigma$.
A constrained
term is a tuple of a term $s$ and a logical constraint $\varphi$ written as
$s~\CO{\varphi}$ (which we call \emph{non-existentially constrained terms} below).
Two non-existentially constrained terms $s~\CO{\varphi}$, $t~\CO{\psi}$ are
equivalent, written as $s~\CO{\varphi} \sim t~\CO{\psi}$, if for all substitutions $\sigma$ with $\sigma \vDash_\xM \varphi$ there
exists a substitution $\gamma$ with $\gamma \vDash_\xM \psi$ such that $s\sigma =
t\gamma$, and vice versa.
Note that we abuse $\sim$ for both existentially constrained terms and non-existentially ones.
The following correspondence follows immediately by definition.

\begin{restatable}{lemma}{LemCorrespondence}
Let $s~\CO{\varphi}$ and $t~\CO{\psi}$ be non-existentially constrained terms,
and
$\CTerm{X}{s}{\vec{x}}{\varphi}$
and $\CTerm{Y}{t}{\vec{y}}{\psi}$ be existentially constrained terms such that
$X = \Var(s) \cap \Var(\varphi)$,
$\SET{\vec{x}} = \Var(\varphi) \setminus\Var(s)$,
$Y = \Var(t) \cap \Var(\psi)$, and
$\SET{\vec{y}} = \Var(\psi) \setminus\Var(t)$.
Then, $s~\CO{\varphi}$ and $t~\CO{\psi}$ 
are equivalent if and only if
so are the existentially constrained terms $\CTerm{X}{s}{\vec{x}}{\varphi}$
and $\CTerm{Y}{t}{\vec{y}}{\psi}$.
\end{restatable}

It is easy to verify the equivalence of non-existentially constrained
terms via equivalence of existentially constrained terms. However, in 
the next section we will see that it is essential to characterize the 
equivalence of existential constraints.

\section{Equivalence of Variable Renamed Existentially Constrained Terms}
\label{ssec:equivalence characterization of ext cterms}
\label{sec:Equivalence of Variable Renamed Constrained Terms}

We start with a sufficient condition for equivalence of 
existentially constrained terms consisting of renamed variants.

\begin{restatable}{lemma}{LemmaIIIxxi}
\label{lem:equivalence by renaming}
Let $\CTerm{X}{s}{\vec{x}}{\varphi}$, $\CTerm{Y}{t}{\vec{y}}{\psi}$ be
existentially constrained terms.
Suppose that there exists a renaming $\delta$
such that $s\delta = t$, $\delta(X) = Y$,
and $\vDash_\xM (\ECO{\vec{x}}{\varphi})\delta \Leftrightarrow (\ECO{\vec{y}}{\psi})$.
Then, $\CTerm{X}{s}{\vec{x}}{\varphi} \sim \CTerm{Y}{t}{\vec{y}}{\psi}$.
\end{restatable}

\begin{proof}
As $\delta$ is a bijection, we can take the inverse renaming by $\delta^{-1}$. 
From the condition,
we have $t = s\delta^{-1}$, $Y = \delta^{-1}(X)$,
and $\vDash_\xM (\ECO{\vec{y}}{\psi})\delta^{-1}  \Leftrightarrow (\ECO{\vec{x}}{\varphi})$.
By symmetricity of the definition of equivalence,
it suffices to prove 
$\CTerm{X}{s}{\vec{x}}{\varphi} \subsetsim \CTerm{\SET{ x\delta \mid x \in X }}{s\delta}{}{(\ECO{\vec{x}}{\varphi})\delta}$.
Fix an $X$-valued substitution $\sigma$ with 
$\sigma \vDash_\xM \ECO{\vec{x}}{\varphi}$. 
Let us define the substitution $\gamma := \sigma \circ \delta^{-1}$.
Then one can show that $\gamma$ is $Y$-valued, 
$\gamma \vDash_\xM \ECO{\vec{y}}{\psi}$ and $s\sigma = t\gamma$. 
\qed
\end{proof}

\begin{example} 
\label{ex:equivalence of renamed terms}
\Cref{lem:equivalence by renaming}
gives that 
$\CTerm{\SET{x}}{\mathsf{f}(x,y)}{}{x \ge 0}
\sim\CTerm{\SET{y}}{\mathsf{f}(y,x)}{}{y \ge 0}$
using a renaming $\delta = \{ x \mapsto y, y \mapsto x \}$.
Also
$\CTerm{\SET{x}}{\mathsf{f}(x,y)}{z}{x = z + z}
\sim\CTerm{\SET{x}}{\mathsf{f}(x,y)}{}{(x \mod 2) = 0}$ holds.
However, \Cref{lem:equivalence by renaming} does 
not imply any of the equivalences
in \Cref{ex:subsumption and equivalence}.
\end{example}

In a special case, our sufficient condition for equivalence
can be extended to a necessary and sufficient condition---this
is our first characterization.

\begin{restatable}{theorem}{ThmCharacterizationOfEquivalenceRsingRenaming}
\label{thm: characterization of equivalence using renaming}
Let $\delta$ be a renaming.
Let $\CTerm{X}{s}{\vec{x}}{\varphi}$, $\CTerm{Y}{t}{\vec{y}}{\psi}$ be
satisfiable existentially constrained terms such that $s\delta = t$.
Then,
$\CTerm{X}{s}{\vec{x}}{\varphi} \sim \CTerm{Y}{t}{\vec{y}}{\psi}$ if and only if
$\delta(X) = Y$ and $\vDash_\xM (\ECO{\vec{x}}{\varphi})\delta \Leftrightarrow (\ECO{\vec{y}}{\psi})$.
\end{restatable}

\begin{proof}
As the \emph{if} direction follows from \Cref{lem:equivalence by renaming},
we concentrate on the \emph{only-if} direction.
Assume that
$\CTerm{X}{s}{\vec{x}}{\varphi} \sim \CTerm{Y}{s\delta}{\vec{y}}{\psi}$.
We first show by contradiction that $\delta(X) = Y$.
Assume that $\delta(X) \ne Y$, hence $X \neq \delta^{-1}(Y)$.
By $Y \subseteq \Var(t)$, we have $\delta^{-1}(Y) \subseteq 
\delta^{-1}(\Var(t)) = \Var(t\delta^{-1}) = \Var(s)$.
Thus, combining with $X \subseteq \Var(s)$, 
we have that $X\cup \delta^{-1}(Y) \subseteq \Var(s)$.
Thus, there exists a variable $z \in \Var(s)$ 
such that $z \notin X$ and $z \in \delta^{-1}(Y)$
or vice versa.

We first consider the case
that $z \notin X$ and $z \in \delta^{-1}(Y)$.
Let $p$ be a position of $z$ in $s$, i.e., $s|_p=z$.
Since $\CTerm{X}{s}{\vec{x}}{\varphi}$ is satisfiable, 
there exists an $X$-valued substitution $\sigma$ with $\sigma \vDash_\xM \ECO{\vec{x}}{\varphi}$
and $\Dom(\sigma) = X$, by \Cref{lem:basic-properties-of-X-valued-substitutions}.
As $z \notin X = \Dom(\sigma)$, it follows
that $z\sigma \notin \Val$.
By our assumption, there
exist a $Y$-valued substitution $\gamma$ with 
$\gamma \vDash_\xM \ECO{\vec{y}}{\psi}$ and $s\sigma = s\delta\gamma$.
As $\gamma$ is $Y$-valued and $\delta(z) \in Y$,
we have $(s\delta\gamma)|_p =  s|_p\delta\gamma = \delta(z)\gamma \in \Val$.
This contradicts $(s\sigma)|_p = z\sigma \notin \Val$ and $s\sigma = s\delta\gamma$.
The case that $z \in X$ and $z \notin \delta^{-1}(Y)$
follows similarly.

Next, we show that 
$\vDash_\xM (\ECO{\vec{x}}{\varphi})\delta \Leftrightarrow (\ECO{\vec{y}}{\psi})$.
To show that
$\vDash_\xM (\ECO{\vec{x}}{\varphi})\delta \Rightarrow (\ECO{\vec{y}}{\psi})$,
let $\rho$ be a valuation such that
$\vDash_{\xM,\rho} (\ECO{\vec{x}}{\varphi})\delta$.
Then, $\vDash_{\xM} (\ECO{\vec{x}}{\varphi})(\rho \circ \delta)$.
Take a substitution $\sigma := (\rho \circ \delta)|_X$.
Then, clearly, $\sigma$ is $X$-valued.
From $\FVar(\ECO{\vec{x}}{\varphi}) \subseteq X$
and $\vDash_{\xM} (\ECO{\vec{x}}{\varphi})\delta\rho$,
it also follows $\sigma \vDash_{\xM} \ECO{\vec{x}}{\varphi}$.
Using our assumption, one obtains a $Y$-valued substitution 
$\gamma$ with $\gamma \vDash_\xM \ECO{\vec{y}}{\psi}$ 
and $s\sigma = t\gamma = s\delta\gamma$.
Thus, we have that $\sigma|_{\Var(s)} = (\gamma \circ \delta)|_{\Var(s)}$.
Hence, $\sigma(\delta^{-1}(x)) = \gamma(x)$ for all $x \in \delta(\Var(s)) = \Var(t)$.
Then, as $\FVar(\ECO{\vec{y}}{\psi}) \subseteq \Var(t)$,
it follows from $\gamma \vDash_\xM \ECO{\vec{y}}{\psi}$ 
that $\vDash_\xM (\ECO{\vec{y}}{\psi})\delta^{-1}\sigma$.
By our choice of $\sigma$,
we have $(\rho \circ \delta)|_X = \sigma|_{X}$.
Thus, $\rho(y) = \sigma(\delta^{-1}(y))$ for all $y \in \delta(X) = Y$.
Hence, by $\FVar(\ECO{\vec{y}}{\psi}) \subseteq Y$,
it follows $\vDash_\xM (\ECO{\vec{y}}{\psi})\rho$.
Thus, $\vDash_{\xM,\rho} \ECO{\vec{y}}{\psi}$.
To show
$\vDash_\xM (\ECO{\vec{y}}{\psi}) \Rightarrow (\ECO{\vec{x}}{\varphi})\delta$,
we apply the proofs so far to the symmetric proposition 
with replacing $\delta$ by $\delta^{-1}$.
\qed
\end{proof}

In~\cite[Section~2.3]{FKN17tocl}, the equivalence of two non-existentially constrained terms 
$s ~ [\varphi]$ and $s ~ [\psi]$
is characterized 
by the validity of $(\ECO{\vec{x}}{\varphi}) \Leftrightarrow (\ECO{\vec{y}}{\psi})$, where $\SET{\vec{x}} = \Var(\varphi) \setminus \Var(s)$ and $\SET{\vec{y}} = \Var(\psi) \setminus \Var(t)$.
However, this characterization 
is not correct:
$(x = x) \Leftrightarrow \m{true}$ is valid, but 
$\m{f}(x)~\CO{x = x} \not\sim \m{f}(x)~\CO{\m{true}}$.
The characterization is recovered and adapted for existentially constrained terms as follows,
using \Cref{thm: characterization of equivalence using renaming}:

\begin{restatable}{corollary}{PropEquivalenceByConstraintValidty}
\label{prop:equivalence-by-constraint-validity}
\label{cor:equivalence-by-constraint-validity}
Let $\CTerm{X}{s}{\vec{x}}{\varphi}$, $\CTerm{Y}{s}{\vec{y}}{\psi}$ be
satisfiable existentially constrained terms. Then,
$\CTerm{X}{s}{\vec{x}}{\varphi} \sim \CTerm{Y}{s}{\vec{y}}{\psi}$ if and only if
$X = Y$ and
$\xM \vDash (\ECO{\vec{x}}{\varphi}) \Leftrightarrow (\ECO{\vec{y}}{\psi})$.
\end{restatable}

\section{Pattern-General Existentially Constrained Terms and Their Equivalence}
\label{ssec:pattern general constrained terms}

In this section, we introduce and focus on a notion for a general form of existentially constrained terms.
A term $s$ is said to be \emph{linear w.r.t.\ a set $X$ of variables} if every
variable in $X$ appears in $s$ at most once.

\begin{definition}[Pattern-General Existentially Constrained Term]
\label{def:pattern-general}
An existentially constrained term $\CTerm{X}{s}{\vec{x}}{\varphi}$ is called
\emph{pattern-general}
if $s$ is $X$-linear and 
$\Val(s) = \varnothing$.
\end{definition}
At this point of the paper, readers may wonder why this condition
is named ``pattern-general''---we 
postpone explaining the reason until
\Cref{lem:characterization-of-pattern-general}.
We call an existentially constrained term $\CTerm{X}{s}{\vec{x}}{\varphi}$ 
\emph{value-free} (cf.~\cite{Kop17}) if $\Val(s) = \varnothing$.
We call $\CTerm{X}{s}{\vec{x}}{\varphi}$ \emph{linear w.r.t.\ logical variables} (LV-linear, for short) if $s$ is $X$-linear.
Thus, $\CTerm{X}{s}{\vec{x}}{\varphi}$ is pattern-general if and only if it is value-free and LV-linear.

\begin{example} 
\label{ex:pattern-general constrained terms}
The term $\CTerm{\SET{x}}{\mathsf{f}(x,y)}{}{x \ge 0}$
is pattern-general, while the terms
$\CTerm{\SET{x}}{\mathsf{f}(x,1)}{}{x \ge 0}$
and 
$\CTerm{\SET{x}}{\mathsf{f}(x,x)}{}{x \ge 0}$ are not.
Note that
$\CTerm{\SET{x}}{\mathsf{f}(y,\mathsf{f}(y,x))}{}{x \ge 0}$
is indeed pattern-general.
\end{example}

In the following we define the transformation from an existentially constrained term into an equivalent one which is in addition pattern-general.

\begin{definition}[PG-transformation]
We define the transformation $\PG$ 
on existentially constrained terms
as follows:
$
\PG(\CTerm{X}{s}{\vec{x}}{\varphi}) = \CTerm{Y}{t}{\vec{y}}{\psi}
$, 
where
    $Y = \SET{w_1,\ldots,w_n}$,
    $t = s[w_1,\ldots,w_n]_{p_1,\ldots,p_n}$, 
    $\SET{\vec{y}} = \SET{\vec{x}} \cup X$,
        and
    $\psi = (\varphi \land \bigwedge_{i=1}^n (s|_{p_i} = w_i))$
with pairwise distinct fresh variables $w_1,\ldots,w_n$
and positions $p_1,\ldots,p_n$ such that $\Pos_{X\cup\Val}(s) = \SET{p_1,\ldots,p_n}$.
\end{definition}

\begin{example}[Cont'd from \Cref{ex:pattern-general constrained terms}]
\label{ex:PG-transformation}
We obtain $\PG(\CTerm{\SET{x}}{\mathsf{f}(x,1)}{}{x \ge 0})
= \CTerm{\SET{w_1,w_2}}{\mathsf{f}(w_1,w_2)}{x}{(x \ge 0) \land (x = w_1) \land (1 = w_2)}$,
as well as 
$\PG(\CTerm{\SET{x}}{\mathsf{f}(x,x)}{}{x \ge 0})
= \CTerm{\SET{w_1,w_2}}{\mathsf{f}(w_1,w_2)}{x}{(x \ge 0) \land (x = w_1) \land (x = w_2)}$.
\end{example}

The following lemma proving well-definedness of the
PG-transformation can be shown in a straightforward manner.

\begin{restatable}{lemma}{LemmaBasicPropertiesOfPG}
\label{lem:basic-properties-of-PG}
Let $\CTerm{X}{s}{\vec{x}}{\varphi}$ be an existentially constrained term.
Then, both of the following statements hold:
\begin{enumerate}
\renewcommand{\labelenumi}{\Bfnum{\arabic{enumi}}}
    \item
    $\PG(\CTerm{X}{s}{\vec{x}}{\varphi})$ is a pattern-general existentially constrained term, 
        and
    \item
    if $\CTerm{X}{s}{\vec{x}}{\varphi}$ is satisfiable, then so is $\PG(\CTerm{X}{s}{\vec{x}}{\varphi})$.
\end{enumerate}
\end{restatable}

The following theorem states the correctness of the PG-transformation.

\begin{restatable}{theorem}{TheoremCorrectnessOfPG}
\label{thm:correctness-of-PG}
Let $\CTerm{X}{s}{\vec{x}}{\varphi}$ be an existentially constrained term.
Then, we have $\CTerm{X}{s}{\vec{x}}{\varphi} \sim \PG(\CTerm{X}{s}{\vec{x}}{\varphi})$.
\end{restatable}

\begin{proof}
Let $\Pos_{X\cup\Val}(s) = \SET{p_1,\ldots,p_n}$ and $\PG(\CTerm{X}{s}{\vec{x}}{\varphi}) = \CTerm{Y}{t}{\vec{y}}{\psi}$
where
$w_1,\ldots,w_n$ are pairwise distinct fresh variables, 
$t = s[w_1,\ldots,w_n]_{p_1,\ldots,p_n}$, 
$Y = \SET{w_1,\ldots,w_n}$,
$\SET{\vec{y}} = \SET{\vec{x}} \cup X$,
        and
$\psi = (\varphi \land \bigwedge_{i=1}^n (s|_{p_i} = w_i))$.
($\subsetsim$)
Let $\sigma$ be an $X$-valued substitution such that 
$\sigma \vDash_{\xM} \ECO{\vec{x}}{\varphi}$. 
Take a substitution $\gamma = \sigma \circ \SET{ w_i \mapsto s|_{p_i} \mid 1  \le i \le n }$.
Then, we can show that $\gamma$ is $Y$-valued,
$\gamma \vDash_{\xM} \ECO{\vec{y}}{\psi}$,
and $t\gamma = s\sigma$.
($\supsetsim$)
Let $\gamma$ be a $Y$-valued substitution such that 
$\gamma \vDash_{\xM} \ECO{\vec{y}}{\psi}$. 
Then, 
there exists $\vec{v} \in \Val^*$
such that $\vDash_{\xM} \psi\kappa\gamma$,
where $\kappa = \SET{\vec{y} \mapsto \vec{v}}$.
Take $\sigma = \gamma \circ \kappa|_{X}$.
Then, one can show $\sigma$ is $X$-valued,
$\sigma \vDash_\xM \ECO{\vec{x}}{\varphi}$,
$s\sigma = t\gamma$.
\qed
\end{proof}

The following two lemmata, whose proofs are straightforward, demonstrate useful
correspondences between constrained terms before and after the PG
transformation.

\begin{restatable}{lemma}{LemInducedValidityOfRepresentativeSubstitutions} 
\label{lem:induced-validity-of-representative-substitutions}
Let $\CTerm{X}{s}{\vec{x}}{\varphi}$ be a satisfiable constrained term.
Suppose that $\PG(\CTerm{X}{s}{\vec{x}}{\varphi}) = (\CTerm{Y}{s[w_1,\ldots,w_n]_{p_1,\ldots,p_n}}{\vec{y}}{\psi})$
and $\sigma = \SET{w_i \mapsto s|_{p_i} \mid 1 \le i \le n}$
where $\SET{p_1,\ldots,p_n}= \Pos_{X\cup\Val}(s)$.
Then $\vDash_\xM (\ECO{\vec{x}}{\varphi}) \Leftrightarrow (\ECO{\vec{y}}{\psi})\sigma$.
\end{restatable}

\begin{restatable}{lemma}{LemmaIIIxix}
\label{lem:X or val subterms equality from constraints}
Let $\CTerm{X}{s}{\vec{x}}{\varphi}$ be a satisfiable constrained term.
Suppose that $\PG(\CTerm{X}{s}{\vec{x}}{\varphi}) = (\CTerm{Y}{s[w_1,\ldots,w_n]_{p_1,\ldots,p_n}}{\vec{y}}{\psi})$
where $\SET{p_1,\ldots,p_n} = \Pos_{X\cup\Val}(s)$.
\begin{enumerate}
\renewcommand{\labelenumi}{\Bfnum{\arabic{enumi}}}
\item 
For any $i,j \in \SET{1,\ldots,n}$, 
$\vDash_\xM (\ECO{\vec{x}}{\varphi}) \Rightarrow (s|_{p_i} = s|_{p_j})$
if and only if 
$\vDash_\xM (\ECO{\vec{y}}{\psi}) \Rightarrow (w_i = w_j)$.
\item
For any $i \in \SET{1,\ldots,n}$ and $v \in \Val$, 
$\vDash_\xM (\ECO{\vec{x}}{\varphi}) \Rightarrow (s|_{p_i} = v)$
if and only if 
$\vDash_\xM (\ECO{\vec{y}}{\psi}) \Rightarrow (w_i = v)$.
\end{enumerate}
\end{restatable}

Next, we characterize equivalence of satisfiable pattern-general existentially constrained terms.
Before the characterization, we show useful fundamental properties of 
equivalence.

\begin{restatable}{lemma}{LemmaBasicDerivedPropertiesOfEquivalence} 
\label{lem:basic-derived-properties-of-equivalence}
Let $\CTerm{X}{s}{\vec{x}}{\varphi}, \CTerm{Y}{t}{\vec{y}}{\psi}$ be satisfiable existentially constrained terms such that $\CTerm{X}{s}{\vec{x}}{\varphi} \sim \CTerm{Y}{t}{\vec{y}}{\psi}$.
Then, all of the following statements hold:
(1) $\Pos(s) = \Pos(t)$,
(2) $\Pos_{X\cup \Val}(s) = \Pos_{Y\cup \Val}(t)$,
(3) $\Pos_{\xV\setminus X}(s) = \Pos_{\xV\setminus Y}(t)$,
(4) $\Pos_{\xF\setminus \Val}(s) = \Pos_{\xF \setminus \Val}(t)$,
(5) $s(p) = t(p)$ for any position $p\in \Pos_{\xF \setminus \Val}(s)$,
(6) for any position $p \in \Pos(s)$, if $s|_p \in X$ and $t|_p \in \Val$, then $\sigma(s|_p) = t|_p$ for any $X$-valued substitution $\sigma$ such that 
    $\sigma \vDash_\xM \ECO{\vec{x}}{\varphi}$, 
    and
(7) there exists a renaming $\theta\colon \Var(s) \setminus X \to \Var(t) \setminus Y$ 
    such that 
    $\theta(s[~]_{p_1,\ldots,p_n}) = t[~]_{p_1,\ldots,p_n}$, 
    where $\SET{p_1,\ldots,p_n} = \Pos_{X \cup \Val}(s)$ ($=\Pos_{Y \cup \Val}(t)$).
\end{restatable}

\begin{proof}
We present the proof of (1);
all remaining cases are proven similarly.
    By \Cref{lem:basic-properties-of-X-valued-substitutions},
    we have $\sigma \vDash_{\xM} \ECO{\vec{x}}{\varphi}$
    and $\Pos(s) = \Pos(s\sigma)$ for some $X$-valued substitution $\sigma$.
    By assumption, we obtain
    a $Y$-valued substitution $\gamma$ such that 
    $\gamma \vDash_\xM \ECO{\vec{y}}{\psi}$
    and $s\sigma = t\gamma$.
    Thus, we have that $\Pos(t) \subseteq \Pos(t\gamma) = \Pos(s\sigma) =  \Pos(s)$.
    Similarly, $\Pos(s) \subseteq \Pos(t)$.
    Therefore, we have that $\Pos(s) = \Pos(t)$.
\qed
\end{proof}

Before proceeding to the actual characterization,
let us revisit why our naming ``pattern-general'' is quite natural.
A term $s$ is called \emph{most general} in a set $S \subseteq \xT(\xF, \xV)$ of terms
whenever for all $t \in S$ there exists a substitution $\sigma$ such that $s\sigma = t$.
We refer to \Cref{sec:missing-proofs} for the proof.

\begin{restatable}{lemma}{LemChacterizationOfPatternGeneral}
\label{lem:characterization-of-pattern-general}
Let $\CTerm{X}{s}{\vec{x}}{\varphi}$ be a satisfiable existentially constrained term.
Then, $\CTerm{X}{s}{\vec{x}}{\varphi}$ is pattern-general if and only if
$s$ is most general in $\SET{t \mid \CTerm{X}{s}{\vec{x}}{\varphi} \sim \CTerm{Y}{t}{\vec{y}}{\psi}}$.
\end{restatable}

\noindent
Note that $\CTerm{\SET{x}}{\m{g}(x)}{}{\m{false}}$
and $\CTerm{\SET{x}}{x}{}{\m{false}}$ are equivalent and pattern-general by definition.
Thus, the lemma does not hold if $\CTerm{X}{s}{\vec{x}}{\varphi}$ is not satisfiable.

We are now ready to present a characterization 
of equivalence of satisfiable pattern-general existentially constrained terms.

\begin{restatable}{theorem}{TheoremIIIxxviii} 
\label{thm:complete characterization of equivalence for most sim-general constrained terms}
Let $\CTerm{X}{s}{\vec{x}}{\varphi}, \CTerm{Y}{t}{\vec{y}}{\psi}$ be 
satisfiable pattern-general existentially constrained terms.
Then, $\CTerm{X}{s}{\vec{x}}{\varphi} \sim \CTerm{Y}{t}{\vec{y}}{\psi}$ if and only if there exists a renaming $\rho$ such that
    $s\rho = t$,
    $\rho(X) = Y$,    
        and
    $\vDash_\xM (\ECO{\vec{x}}{\varphi})\rho \Leftrightarrow (\ECO{\vec{y}}{\psi})$.
\end{restatable}

\begin{proof}
The \textit{if} part follows from \Cref{lem:equivalence by renaming}.
We show the \textit{only-if} part.
By our assumption,
no value appears in $s$ or $t$ and hence $\Pos_{\Val}(s) = \Pos_{\Val}(t) = \varnothing$.
Using \Cref{lem:basic-derived-properties-of-equivalence},
let $\SET{p_1,\ldots,p_n} = \Pos_{X}(s) =\Pos_{Y}(t)$,
and $s|_{p_i} = z_i$ and $t|_{p_i} = w_i$ for $1 \leqslant i \leqslant n$.
As $\CTerm{X}{s}{\vec{x}}{\varphi}$ and $\CTerm{Y}{t}{\vec{y}}{\psi}$ are pattern-general,
$z_1,\ldots,z_n$ are pairwise distinct variables, and so are $w_1,\ldots,w_n$.
Let $\rho = \theta \cup \SET{ z_i \mapsto w_i \mid 1 \leq i \leq n}$.
Since $\theta$ is a renaming,
$\SET{z_1,\ldots,z_n} \cap (\Var(s) \setminus X) = \varnothing$, and
$\SET{w_1,\ldots,w_n} \cap (\Var(t) \setminus Y) = \varnothing$,
we know $\rho$ is well-defined and renaming.
Furthermore, $s\rho = (s[z_1,\ldots,z_n]_{p_1,\ldots,p_n})\rho 
= s\rho[z_1\rho,\ldots,z_n\rho]_{p_1,\ldots,p_n} 
= s\theta[w_1,\ldots,w_n]_{p_1,\ldots,p_n} 
= t[w_1,\ldots,w_n]_{p_1,\ldots,p_n} = t$.
The rest of the statements follow 
from this by \Cref{thm: characterization of equivalence using renaming}.
\qed
\end{proof}

\begin{example} 
\label{ex:equivalence of pattern-general terms}
Recall the equivalent existentially constrained terms 
of \Cref{ex:subsumption and equivalence}.
By applying the PG-transformation to each them, we obtain the following:
\[
\begin{array}{l}
\PG(\CTerm{\varnothing}{\mathsf{f}(1,1)}{}{}) 
= \CTerm{\SET{w_1,w_2}}{\mathsf{f}(w_1,w_2)}{}{(1 = w_1) \land (1 = w_2})\\
\PG(\CTerm{\SET{x,y}}{\mathsf{f}(x,y)}{}{(x = y) \land (x = 1)})\\
\quad = \CTerm{\SET{w_1,w_2}}{\mathsf{f}(w_1,w_2)}{}{(x = y) \land (x = 1) \land (x = w_1) \land (y = w_2})\\
\PG(\CTerm{\SET{y}}{\mathsf{f}(1,y)}{}{y = 1})\\
\quad = \CTerm{\SET{w_1,w_2}}{\mathsf{f}(w_1,w_2)}{y}{(y = 1) \land (1 = w_1) \land (y = w_2})\\
\PG(\CTerm{\SET{x}}{\mathsf{f}(x,x)}{}{x = 1})\\
\quad = \CTerm{\SET{w_1,w_2}}{\mathsf{f}(w_1,w_2)}{x}{(x = 1) \land (x = w_1) \land (x = w_2})
\end{array}
\]
It becomes trivial to see the equivalence of these terms by
\Cref{thm:complete characterization of equivalence for most sim-general constrained terms},
as all of these constraints are logically equivalent.
\end{example}

Together with \Cref{thm:correctness-of-PG},
we have a sound and complete characterization of equivalence
of satisfiable existentially constrained terms.

\begin{corollary}
\label{cor:sound and complete characterization of equivalence for PG terms}
Let $\CTerm{X}{s}{\vec{x}}{\varphi}, \CTerm{Y}{t}{\vec{y}}{\psi}$ be 
satisfiable existentially constrained terms.
Then, $\CTerm{X}{s}{\vec{x}}{\varphi} \sim \CTerm{Y}{t}{\vec{y}}{\psi}$ if and only if 
there exists a renaming $\rho$ such that
    $s'\rho = t'$,
    $\rho(X') = Y'$,    
        and
    $\vDash_\xM (\ECO{\pvec{x}'}{\varphi'})\rho \Leftrightarrow (\ECO{\pvec{y}'}{\psi'})$,
where $\PG(\CTerm{X}{s}{\vec{x}}{\varphi}) = \CTerm{X'}{s'}{\pvec{x}'}{\varphi'})$
and $\PG(\CTerm{Y}{t}{\vec{y}}{\psi}) = \CTerm{Y'}{t'}{\pvec{y}'}{\psi'})$.
\end{corollary}

Transformations of non-existentially constrained terms into \emph{value-free} ones can be seen in~\cite{Kop17,KN24jip}.
The PG-transformation in this section is an extension to that of existentially constrained terms into pattern-general ones (i.e., not only value-free but also LV-linear ones).

\section{General Characterization of Equivalence}
\label{sec:general-characterization-of-equivalence}

To check the equivalence of existentially constrained terms, it is often
convenient to apply the PG transformation and
use \Cref{thm:complete characterization of equivalence for most sim-general constrained terms}. However, in a theoretical analysis, one frequently needs to handle the
equivalence of non-pattern-general terms as well, where the PG transformation is
less suitable. Therefore, a general criterion for the equivalence of
existentially constrained terms is preferred.

In this section, we generalize \Cref{thm:complete characterization of
equivalence for most sim-general constrained terms} to arbitrary existentially
constrained terms, i.e., we characterize equivalence of existentially
constrained terms which are not assumed to be pattern-general. To this end, we
introduce new notions and notations.

\begin{definition} 
Let $\CTerm{X}{s}{\vec{x}}{\varphi}$ be an existentially constrained term.
Then, we define a binary relation $\sim_{\Pos_{X\cup \Val}(s)}$ over the positions in $\Pos_{X\cup \Val}(s)$ as follows:
$p \sim_{\Pos_{X\cup \Val}(s)} q$
if and only if
$\vDash_\xM ((\ECO{\vec{x}}{\varphi}) \Rightarrow s|_p = s|_q)$.
\end{definition}
We may omit the subscript $\Pos_{X\cup\Val}(s)$ of $\sim_{\Pos_{X\cup \Val}(s)}$ if it is clear from the context.

\begin{example}
\label{ex:sim over positions of X U Val}
Consider an existentially constrained term
$\CTerm{X}{s}{}{\varphi}$
where $X = \SET{x,x',y,y'}$, 
$s = \m{h}(x,x',0,y,y,y',0 \times 10)$ with $\m{h} \in \xFTe$, 
and $\varphi = (x = x') \land (x' = 0) \land (y = y')$.
Then, $\Pos_{X\cup \Val}(s) = \SET{1,2,3,4,5,6,7.1,7.2}$
and we have $1 \sim 2 \sim 3 \sim 7.1$ and $4 \sim 5 \sim 6$.
\end{example}

The proof of the following lemma is straightforward.

\begin{restatable}{lemma}{LemPropertiesOfEquivalenceOverXvarValPositions} 
\label{lem:properties-of-equivalence-over-X-var-val-positions}
Let $\CTerm{X}{s}{\vec{x}}{\varphi}$ be an existentially constrained term. Then,
\begin{enumerate}
\renewcommand{\labelenumi}{\Bfnum{\arabic{enumi}}}
    \item $\sim_{\Pos_{X\cup \Val}(s)}$ is an equivalence relation over the positions in $\Pos_{X\cup \Val}(s)$,
        and
    \item for any positions $p,q \in \Pos_{X\cup \Val}(s)$, if $s|_p = s|_q$, then $p \sim_{\Pos_{X\cup \Val}(s)} q$.
\end{enumerate}
\end{restatable}

The equivalence class of a position $p \in \Pos_{X\cup\Val}(s)$ w.r.t.\ 
$\sim_{\Pos_{X\cup \Val}(s)}$ is denoted by $[p]_{\sim_{\Pos_{X\cup \Val}(s)}}$.
If it is clear from the context then we may simply denote it by $[p]_\sim$.
We further denote the representative of $[p]_\sim$ by $\hat{p}$.

\begin{definition} 
Let $\CTerm{X}{s}{\vec{x}}{\varphi}$ be a satisfiable existentially constrained term.
Let $\Pos_{\Val!}(s) = \{ p \in \Pos_{X\cup\Val}(s) \mid$  
there exists $v \in \Val$ such that
$\vDash_\xM ((\ECO{\vec{x}}{\varphi}) \Rightarrow (s|_p = v)) \}$.
\end{definition}

\begin{example}[Cont'd from \Cref{ex:sim over positions of X U Val}]
\label{ex:presentative of sim-equivalence class over positions of X U Val}
Since $\sim$ is an equivalence relation, we have a quotient set
$\Pos_{X\cup \Val}(s)/{\sim} = \SET{\SET{1,2,3,7.1},\SET{4,5,6},\SET{7.2}}$.
Let $1$ and $4$ be representatives of the equivalence classes
$\SET{1,2,3,7.1}$ and $\SET{4,5,6}$, respectively.
Then, for instance $\hat 3 = 1$ and $\hat 5 = 4$.
Furthermore, we have $\Pos_{\Val!}(s) = \SET{1,2,3,7.1,7.2}$.
\end{example}

It is easy to see the set $\Pos_{\Val!}(s)$ is consistent with 
$\sim_{\Pos_{X\cup \Val}(s)}$ as shown in the following lemma.

\begin{restatable}{lemma}{LemPropertiesOfPosExcl} 
\label{lem:properties-of-PosValExcl}
For any satisfiable existentially constrained term $\CTerm{X}{s}{\vec{x}}{\varphi}$, all of the following statements hold:
\begin{enumerate}
\renewcommand{\labelenumi}{\Bfnum{\arabic{enumi}}}
   \item For each position $p \in \Pos_{\Val!}(s)$, there exists a unique value $v \in \Val$ such that
    $\vDash_\xM ((\ECO{\vec{x}}{\varphi}) \Rightarrow (s|_p = v))$,
    
    \item for any positions $p,q \in \Pos_{X\cup\Val}(s)$, 
    if $p \in \Pos_{\Val!}(s)$ and $p \sim q$, then $q \in \Pos_{\Val!}(s)$,
        and
  \item 
    for any positions $p,q \in \Pos_{X\cup\Val}(s)$ and values $v,v' \in \Val$, 
    if $p \sim q$, $\vDash_\xM ((\ECO{\vec{x}}{\varphi}) \Rightarrow (s|_p = v))$, 
    and $\vDash_\xM ((\ECO{\vec{x}}{\varphi}) \Rightarrow (s|_q = v'))$,  then $v = v'$.
\end{enumerate}
\end{restatable}

As mechanism to instantiate the term part of an existentially constrained term by
values under the equivalence, we introduce representative
substitutions.

\begin{definition} 
\label{def:representative substituions}
Let $\CTerm{X}{s}{\vec{x}}{\varphi}$ be a satisfiable existentially constrained term.
\begin{enumerate}
\renewcommand{\labelenumi}{\Bfnum{\arabic{enumi}}}
\item For each $p \in \Pos_{\Val!}(s)$,  there exists a unique value $v$
    such that $\xM \vDash (\ECO{\vec{x}}{\varphi}) \Rightarrow (s|_p = v)$
    by \Cref{lem:properties-of-PosValExcl}.
    We denote such $v$ by $\Val!(p)$.
\item We define a 
\emph{representative substitution} $\mu_{X}: X \to X \cup \Val$ of 
$\CTerm{X}{s}{\vec{x}}{\varphi}$ as follows:
\[
\mu_{X}(z) = 
\left\{
\begin{array}{ll}
\Val!(p) & \mbox{if $s(p) = z$ for some $p \in \Pos_{\Val!}(s)$},\\
s(\hat{p}) & \mbox{otherwise},
\end{array}
\right.
\]
where $\hat{p}$ is the representative of the equivalence class $[p]_\sim$.
Here note that if $s(p) = z$ for some $p \in \Pos_{\Val!}(s)$,
then $q \in \Pos_{\Val!}(s)$ for any $q$ such that $s(q) = z$
by \Cref{lem:properties-of-PosValExcl}~\Bfnum{2}.

\item
Provided it is clear from the contexts, 
for simplicity, we shorten $\mu_{X}(x)$ by $\hat{x}$ and so does $\mu_{X}(X)$ by $\hat{X}$.
We also put $\hat{v} = v$ for $v \in \Val$.
\end{enumerate}
\end{definition}

\begin{example}[Cont'd from \Cref{ex:presentative of sim-equivalence class over positions of X U Val}]
\label{ex:representative substituions}
We have $\Val!(p) = 0 \in \Val$ for $p \in \SET{1,2,3,7.1}$,
and $\Val!(7.2) = 10 \in \Val$.
Since $X = \SET{x,x',y,y'}$, the representative substitution is
$\mu_{X} = \SET{ x \mapsto 0, x' \mapsto 0, y \mapsto y, y' \mapsto y }$
and $\hat{X} = \SET{1,y}$.
\end{example}

\begin{example}[Cont'd from \Cref{ex:representative substituions}]
\label{ex:equivalence for general case II}
Let $\CTerm{Y}{t}{}{\psi}$ be another existentially constrained term
such that $Y=\SET{x,x',y,z}$, $t = \m{h}(z,z,z,x,x',x,z \times y)$, 
and $\psi = ((x \le x') \land (x' \le x) \land (z = 0) \land (y = (z + 2) \times 5))$.
Then, we have $\CTerm{X}{s}{}{\varphi} \sim \CTerm{Y}{t}{}{\psi}$.
Let us see how this equivalence can be observed using our notions.
We have 
$\Pos_{X\cup \Val}(s) = \SET{1,2,3,4,5,6,7.1,7.2} = \Pos_{X\cup \Val}(t)$,
${\sim}_{\Pos_{X\cup \Val}(s)} = {\sim}_{\Pos_{Y\cup \Val}(t)}$,
and $\Pos_{\Val!}(s) = \Pos_{\Val!}(t)$.
Let $1$ and $4$ be representatives of $\SET{1,2,3,7.2}$ and $\SET{4,5,6}$.
Then, we have 
$\mu_{X} = \SET{ x \mapsto 0, x' \mapsto 0, y \mapsto y, y' \mapsto y }$
and
$\mu_{Y} = \SET{ x \mapsto x, x' \mapsto x, y \mapsto 10, z \mapsto 0 }$.
We also have 
$\varphi\mu_X = ((0 = 0) \land (0 = 0) \land (y = y))$ and
$\psi\mu_Y = ((x \le x) \land (x \le x) \land (0 = 0) \land (10 = (0 + 2) \times 5))$.
The remaining key is the correspondence
between the sets $\hat X \cap X = \{ y \}$ and $\hat Y \cap Y = \{ z \}$.
\end{example}

\begin{restatable}{lemma}{LemInducedEqualitiesOfExpliciitSubstitutions} 
\label{lem:induced-equalities-of-expliciit-substitutions}
Let $\CTerm{X}{s}{\vec{x}}{\varphi}$ be a satisfiable 
existentially constrained term, and $\mu_{X}$ a representative substitution of $\CTerm{X}{s}{\vec{x}}{\varphi}$.
Then, $\vDash_\xM (\ECO{\vec{x}}{\varphi}) 
\Rightarrow (z = \mu_{X}(z))$ for any variable $z \in X$.
\end{restatable}

Let $\theta \subseteq X\times Y$ be a binary relation,
and $\tilde{X} \subseteq X$, and $\tilde{Y} \subseteq Y$.
Then, $\theta|_{\tilde{X}} =  \SET{ \langle x, y \rangle \in \theta  \mid x \in \tilde{X} }$
is denoted by $\theta|_{\tilde{X}}: \tilde{X} \to \tilde{Y}$
and used as a function
if it is a function from $\tilde{X}$ to $\tilde{Y}$,
i.e., for any $x \in \tilde{X}$ there exists a unique $y \in \tilde{Y}$
such that $\langle x, y \rangle \in \theta$.
Note that if $\tilde{X} \subseteq \xV$ and $\tilde{Y} \subseteq \xT(\xF,\xV)$, then $\theta|_{\tilde{X}}$ is a substitution.

We now present key properties required to establish our general characterization.

\begin{restatable}{lemma}{LemInducedPropertisOfEuivalence} 
\label{lem:induced-propertis-of-equivalence}
Let $\CTerm{X}{s}{\vec{x}}{\varphi}, \CTerm{Y}{t}{\vec{y}}{\psi}$ be satisfiable existentially constrained terms such that $\CTerm{X}{s}{\vec{x}}{\varphi} \sim \CTerm{Y}{t}{\vec{y}}{\psi}$.
Let $\Pos_{X\cup\Val}(s) = \Pos_{Y\cup\Val}(t) = \SET{p_1,\ldots,p_n}$.
Then, all of the following statements hold:
\begin{enumerate}
\renewcommand{\labelenumi}{\Bfnum{\arabic{enumi}}}
    \item for any $i,j \in \SET{1,\ldots,n}$, 
        $\xM \vDash (\ECO{\vec{x}}{\varphi}) \Rightarrow (s|_{p_i} = s|_{p_j})$
        if and only if  
        $\xM \vDash (\ECO{\vec{y}}{\psi}) \Rightarrow (t|_{p_i} = t|_{p_j})$,
    \item for any $i \in \SET{1,\ldots,n}$ and $v \in \Val$, 
        $\xM \vDash (\ECO{\vec{x}}{\varphi}) \Rightarrow (s|_{p_i} = v)$
        if and only if 
        $\xM \vDash (\ECO{\vec{y}}{\psi}) \Rightarrow (t|_{p_i} = v)$, and
    \item  $\xM \vDash (\ECO{\vec{x}}{\varphi})\mu_X\theta \Leftrightarrow (\ECO{\vec{y}}{\psi})\mu_Y$, 
    let ${\sim} = {\sim}_{\Pos_{X\cup \Val}(s)} = {\sim}_{\Pos_{X\cup \Val}(t)}$, 
    and $\mu_X,\mu_Y$ be representative substitutions of
     $\CTerm{X}{s}{\vec{x}}{\varphi}$ and $\CTerm{Y}{t}{\vec{y}}{\psi}$, respectively, 
    based on the same representative for
    each equivalence class $[p_i]_\sim$ ($1 \le i \le n$), and
    we have 
    $\vDash_\xM (\ECO{\vec{x}}{\varphi})\mu_X\theta|_{\tilde{X}}
    \Leftrightarrow (\ECO{\vec{y}}{\psi})\mu_Y$
    with a renaming $\theta|_{\tilde{X}} \colon \tilde{X} \to \tilde{Y}$,
    where $\theta =  \SET{ \langle s|_{p_i}, t|_{p_i} \rangle \mid 1 \le i \le n }$,
    $\tilde{X} = \hat{X} \cap X$, and $\tilde{Y} = \hat{Y}\cap Y$.
\end{enumerate}
\end{restatable}

\begin{proof}
Let $\PG(\CTerm{X}{s}{\vec{x}}{\varphi}) 
= (\CTerm{W}{s[w_1,\ldots,w_n]_{p_1,\ldots,p_n}}{\pvec{x}'}{\varphi'})$
and
$\PG(\CTerm{Y}{t}{\vec{y}}{\psi}) 
= (\CTerm{W'}{t[w'_1,\ldots,w'_n]_{p_1,\ldots,p_n}}{\pvec{y}'}{\psi'})$, 
where we let
$W = \SET{w_1,\ldots,w_n}$,
$\SET{\pvec{x}'} = \SET{\vec{x}}\cup X$,
$\varphi' = (\varphi \land \bigwedge_{i=1}^n (w_i = s|_{p_i}))$,
$W' = \SET{w_1',\ldots,w_n'}$,
$\SET{\pvec{y}'} = \SET{\vec{y}}\cup Y$,
and
$\psi' = (\psi \land \bigwedge_{i=1}^n (w'_i = t|_{p_i}))$.
Then, by our assumption and \Cref{thm:correctness-of-PG}, we have
\[
(\CTerm{W}{s[w_1,\ldots,w_n]_{p_1,\ldots,p_n}}{\pvec{x}'}{\varphi'})
\sim
(\CTerm{W}{t[w'_1,\ldots,w'_n]_{p_1,\ldots,p_n}}{\pvec{y}'}{\psi'}). 
\]
By \Cref{thm:complete characterization of equivalence for most sim-general constrained terms}, 
there is a renaming $\delta$ such that 
$s[w_1,\ldots,w_n]\delta = t[w_1',\ldots,w_n']$
and $\vDash_\xM (\ECO{\pvec{x}'}{\varphi'})\delta \Leftrightarrow \ECO{\pvec{y}'}{\psi'}$.
Thus, $\delta(w_i) = w_i'$. 
Then, using \Cref{lem:X or val subterms equality from constraints}~\Bfnum{1},
\[
\begin{array}{lcl}
 \vDash_\xM (\ECO{\vec{x}}{\varphi}) \Rightarrow (s|_{p_i} = s|_{p_j})
&\iff& \vDash_\xM (\ECO{\pvec{x}'}{\varphi'}) \Rightarrow (w_i = w_j)\\
&\iff& \vDash_\xM (\ECO{\pvec{x}'}{\varphi'})\delta \Rightarrow (w_i\delta = w_j\delta) \\
&\iff& \vDash_\xM (\ECO{\pvec{y}'}{\psi'}) \Rightarrow (w_i' = w_j')\\
&\iff& \vDash_\xM (\ECO{\pvec{y}}{\psi}) \Rightarrow (t|_{p_i} = t|_{p_j})\\
\end{array}
\]
This proves~\Bfnum{1}. 
\Bfnum{2} follows similarly, by
using \Cref{lem:X or val subterms equality from constraints}~\Bfnum{2}.
In order to show~\Bfnum{3}, take $\theta = \SET{ \langle s|_{p_i}, t|_{p_i} \rangle \mid 1 \le i \le n }$.
Then, by~\Bfnum{1},\,\Bfnum{2}, and our assumption, we have
\begin{itemize}
\item $\mu_X(s|_{p_i}) \in \Val$ if and only if  $\mu_Y(t|_{p_i}) \in \Val$, and moreover, 
\item $s|_{p_i} \in \tilde{X}$ if and only if $p_i$ is a representative of an equivalence class $[p_i]$
if and only if $t|_{p_i} \in \tilde{Y}$.
\end{itemize}
Thus, 
$\{ \langle s|_{p_i}, t|_{p_i} \rangle \mid
p_i$ is a representative of an equivalence class 
$[p_i]$ such that $\mu_X(s|_{p_i}) \notin \Val \} = \theta|_{\tilde{X}}$
is a renaming from $\tilde{X}$ to $\tilde{Y}$.
Take an arbitrary but fixed function $\theta': X \to Y$
such that $\theta'(s|_{p_i}) \in \SET{ t|_{p_j} \mid p_i \sim p_j }$.
Then (i) $\theta|_{\tilde{X}}(\mu_X(s|_{p_i})) = \theta|_{\tilde{X}}(s|_{\hat{p_i}}))) 
= t|_{\hat{p_i}} = \mu_Y(t|_{p_j}) = \mu_Y(\theta'(s|_{p_i}))$.
Let $\sigma = \SET{ w_1 \mapsto s|_{p_1},\ldots, w_n \mapsto s|_{p_n}}$
and
$\sigma' = \SET{ w'_1 \mapsto t|_{p_1},\ldots, w'_n \mapsto t|_{p_n}}$.
Then, we have
(ii) $\mu_Y(\theta'(\sigma(w_i)))
= \mu_Y(\theta'(s|_{p_i}))
= \mu_Y(t|_{p_j})
= t|_{\hat{p_i}}
= \mu_Y(t|_{p_i})
= \mu_Y(\sigma'(w_i'))
= \mu_Y(\sigma'(\delta(w_i)))$.
Thus, for any valuation $\rho$,
\[
\begin{array}{lcl@{\qquad}l}
 \vDash_{\xM,\rho} (\ECO{\vec{x}}{\varphi})\mu_X\theta|_{\tilde{X}}
&\iff& \vDash_{\xM,\rho} (\ECO{\pvec{x}'}{\varphi'})\sigma\mu_X\theta|_{\tilde{X}}
& \mbox{by \Cref{lem:induced-validity-of-representative-substitutions}}\\
&\iff& \vDash_{\xM,\rho} (\ECO{\pvec{x}'}{\varphi'})\sigma\theta'\mu_Y
& \mbox{by (i)}\\
&\iff& \vDash_{\xM,\rho}(\ECO{\pvec{x}'}{\varphi'})\delta\sigma'\mu_Y
& \mbox{by (ii)}\\
&\iff& \vDash_{\xM,\rho} (\ECO{\pvec{y}'}{\psi'})\sigma'\mu_Y
& \mbox{by \Cref{lem:induced-validity-of-representative-substitutions}}\\
&\iff& \vDash_{\xM,\rho} (\ECO{\pvec{y}}{\psi}) \mu_Y
\end{array}
\]
Thus, 
$\vDash_{\xM} (\ECO{\vec{x}}{\varphi})\mu_X\theta|_{\tilde{X}} \Leftrightarrow (\ECO{\pvec{y}}{\psi}) \mu_Y$.
\qed
\end{proof}

We finally arrive at a general characterization of equivalence.

\begin{restatable}{theorem}{ThmCompleteCharacterizatinOfEquivalenceOfConstrainedTerms} 
\label{thm:complete characterizatin of equivalence of constrained terms}
Let $\CTerm{X}{s}{\vec{x}}{\varphi}, \CTerm{Y}{t}{\vec{y}}{\psi}$ be satisfiable existentially constrained terms, and $\Pos_{X\cup\Val}(s) = \SET{p_1,\ldots,p_n}$.
Then, $\CTerm{X}{s}{\vec{x}}{\varphi} \sim \CTerm{Y}{t}{\vec{y}}{\psi}$
if and only if
the following statements hold:
\begin{enumerate}
\renewcommand{\labelenumi}{\Bfnum{\arabic{enumi}}}
    \item $\Pos_{X\cup\Val}(s) = \Pos_{Y\cup\Val}(t)$,
    \item there exists a renaming $\rho \colon \Var(s) \setminus X \to \Var(t) \setminus Y$ such that 
    $\rho(s[\,]_{p_1,\ldots,p_n}) = t[\,]_{p_1,\ldots,p_n}$,
    \item for any $i,j \in \SET{1,\ldots,n}$, 
        $\vDash_\xM (\ECO{\vec{x}}{\varphi}) \Rightarrow (s|_{p_i} = s|_{p_j})$
        if and only if 
        $\vDash_\xM (\ECO{\vec{y}}{\psi}) \Rightarrow (t|_{p_i} = t|_{p_j})$,
    \item for any $i \in \SET{1,\ldots,n}$ and $v \in \Val$, 
        $\vDash_\xM (\ECO{\vec{x}}{\varphi}) \Rightarrow (s|_{p_i} = v)$
        if and only if 
        $\vDash_\xM (\ECO{\vec{y}}{\psi}) \Rightarrow (t|_{p_i} = v)$,
        and
    \item
    let ${\sim} = {\sim}_{\Pos_{X\cup \Val}(s)} = {\sim}_{\Pos_{X\cup \Val}(t)}$, 
    and $\mu_X,\mu_Y$ be representative substitutions of
     $\CTerm{X}{s}{\vec{x}}{\varphi}$ and $\CTerm{Y}{t}{\vec{y}}{\psi}$, respectively, 
    based on the same representative for
    each equivalence class $[p_i]_\sim$ ($1 \le i \le n$), and
    we have 
    $\vDash_\xM (\ECO{\vec{x}}{\varphi})\mu_X\theta|_{\tilde{X}}
    \Leftrightarrow (\ECO{\vec{y}}{\psi})\mu_Y$
    with a renaming $\theta|_{\tilde{X}} \colon \tilde{X} \to \tilde{Y}$,
    where $\theta =  \SET{ \langle s|_{p_i}, t|_{p_i} \rangle \mid 1 \le i \le n }$,
    $\tilde{X} = \hat{X} \cap X$, and $\tilde{Y} = \hat{Y}\cap Y$.
\end{enumerate}
\end{restatable}

\begin{proof}[Sketch]
The \textit{only-if} part follows \Cref{lem:basic-derived-properties-of-equivalence} 
and \Cref{lem:induced-propertis-of-equivalence}.
The \textit{if} part is shown by applying the PG-transformation
to $\CTerm{X}{s}{\vec{x}}{\varphi}$ and $\CTerm{Y}{t}{\vec{y}}{\psi}$,
and using \Cref{thm:complete characterization of equivalence for most sim-general constrained terms}.
For this 
\Cref{lem:induced-equalities-of-expliciit-substitutions}
and 
\Cref{lem:induced-propertis-of-equivalence}
are needed.
\qed
\end{proof}

\section{Conclusion}
\label{sec:conclusion}

In this paper, we introduced the notion of existentially constrained terms
and provided characterizations of their equivalence. We showed
that original formalization of constrained terms can be embedded into
existentially constrained terms, and that equivalence between constrained
terms coincides with the equivalence of their embedded counterparts.

We first presented a characterization of equivalent existentially
constrained terms in the case where the term part of one is a renaming of
the other. Then, we introduced pattern-general existentially constrained terms
and the PG-transformation, demonstrating that the
PG-transformation maps each existentially constrained term to an equivalent
pattern-general one.

We have shown a characterization of equivalent pattern-general existentially
constrained terms, which in turn yields a sound and complete
characterization of equivalent existentially constrained terms via the
PG-transformation. Finally, we also established a sound and complete
characterization of equivalence for existentially constrained terms in general,
independent of the notion of pattern-general existentially constrained
terms.

Our future work is to apply our results on existentially constrained terms to other areas of the LCTRS formalism.
This includes investigating their benefits for automated reasoning by LCTRS tools.

\bibliographystyle{splncs04}
\bibliography{biblio}

@string{proc = "Proceedings of the "}

@string{cade = " CADE"}

@string{IJCAR = " IJCAR"}

@string{lipics = "LIPIcs"}

@string{fscd = " FSCD"}

@string{lopstr = " LOPSTR"}

@string{vstte= " VSTTE"}

@string{aplas= " APLAS"}

@string{frocos= " FroCoS"}

@string{lnai = "Lecture Notes in Artificial Intelligence"}

@string{lncs = "Lecture Notes in Computer Science"}

@book{BN98,
  title = {Term Rewriting and All That},
  author = {Franz Baader and Tobias Nipkow},
  year = {1998},
  doi = {10.1145/505863.505888},
  publisher = {Cambridge University Press},
}

@book{O02,
  title = {Advanced Topics in Term Rewriting},
  author = {Enno Ohlebusch},
  year = {2002},
  doi = {10.1007/978-1-4757-3661-8},
  publisher = {Springer},
}

@article{FKN17tocl,
 author = {Carsten Fuhs and Cynthia Kop and Naoki Nishida},
 title = {Verifying Procedural Programs via Constrained Rewriting Induction},
 journal = {ACM Transactions on Computational Logic},
 volume = {18},
 number = {2},
 pages = {14:1--14:50},
 year = {2017},
 doi = {10.1145/3060143},
}

@article{K16,
 author    = "Cynthia Kop",
 title     = "Termination of {LCTRSs}",
 journal   = "CoRR",
 volume    = "abs/1601.03206",
 year      = 2016,
 doi       = "10.48550/ARXIV.1601.03206"
}

@inproceedings{WM21,
 author    = "Sarah Winkler and Georg Moser",
 editor    = "Maribel Fern{\'a}ndez",
 title     = "Runtime Complexity Analysis of Logically Constrained
              Rewriting",
 booktitle = proc # "30th" # lopstr,
 series    = lncs,
 volume    = 12561,
 pages     = "37--55",
 year      = 2021,
 doi       = "10.1007/978-3-030-68446-4_2"
}

@inproceedings{SM23,
 author    = "Jonas Sch{\"o}pf and Aart Middeldorp",
 editor    = "Brigitte Pientka and Cesare Tinelli",
 title     = "Confluence Criteria for Logically Constrained Rewrite
              Systems",
 booktitle = proc # "29th" # cade,
 series    = lnai,
 volume    = "14132",
 pages     = "474--490",
 year      = 2023,
 doi       = "10.1007/978-3-031-38499-8_27"
}

@inproceedings{WM18,
 author    = "Sarah Winkler and Aart Middeldorp",
 editor    = "H\'{e}l\`{e}ne Kirchner",
 title     = "Completion for Logically Constrained Rewriting",
 booktitle = proc # "3rd" # fscd,
 series    = lipics,
 volume    = 108,
 pages     = "30:1--30:18",
 year      = 2018,
 doi       = "10.4230/LIPIcs.FSCD.2018.30"
}

@inproceedings{KN14aplas,
 author = {Cynthia Kop and Naoki Nishida},
 editor = "Jacques Garrigue",
 title = {Automatic Constrained Rewriting Induction towards Verifying Procedural Programs},
 booktitle = proc # "12th" # aplas,
 series = lncs,
 volume = {8858},
 pages = {334--353},
 year = {2014},
 doi = {10.1007/978-3-319-12736-1_18},
}

@inproceedings{KN13frocos,
 author = {Cynthia Kop and Naoki Nishida},
 editor = {Pascal Fontaine and Christophe Ringeissen and Renate A. Schmidt},
 title = {Term Rewriting with Logical Constraints},
 booktitle = proc # "9th" # frocos,
 series = lncs,
 volume = {8152},
 pages = {343--358},
 year = {2013},
 doi = {10.1007/978-3-642-40885-4_24},
}

@inproceedings{SMM24,
 author    = "Jonas Sch{\"o}pf and  Fabian Mitterwallner and Aart Middeldorp",
 editor    = "Christoph Benzmüller and Marijn J.H. Heule and Renate A. Schmidt",
 title     = "Confluence of Logically Constrained Rewrite Systems Revisited",
 booktitle = proc # "12th" # ijcar,
 series    = lnai,
 volume    = "14740",
 pages     = "298--316",
 year      = 2024,
 doi       = "10.1007/978-3-031-63501-4_16"
}

@inproceedings{NW18,
 author    = "Naoki Nishida and Sarah Winkler",
 editor    = "Ruzica Piskac and Philipp R{\"u}mmer",
 title     = "Loop Detection by Logically Constrained Term Rewriting",
 booktitle = proc # "10th" # vstte,
 series    = lncs,
 volume    = 11294,
 pages     = "309--321",
 year      = 2018,
 doi       = "10.1007/978-3-030-03592-1_18"
}

@inproceedings{ANS24,
 author    = "Takahito Aoto and Naoki Nishida and Jonas Sch{\"o}pf",
 editor    = "Jakob Rehof",
 title     = "Equational Theories and Validity for Logically Constrained
              Term Rewriting",
 booktitle = proc # "9th" # fscd,
 series    = lipics,
 volume    = 299,
 pages     = "31:1--31:21",
 year      = 2024,
 doi       = "10.4230/LIPIcs.FSCD.2024.31"
}

@article{Kop17,
  author       = {Cynthia Kop},
  title        = {Quasi-reductivity of Logically Constrained Term Rewriting Systems},
  journal      = {CoRR},
  volume       = {abs/1702.02397},
  year         = {2017},
  url          = {http://arxiv.org/abs/1702.02397},
  xeprinttype    = {arXiv},
  xeprint       = {1702.02397},
}

@ARTICLE{KN24jip,
  author = {Misaki Kojima and Naoki Nishida},
  title = {On Representations of Waiting Queues for Semaphores in Logically Constrained Term Rewrite Systems with Constant Destinations},
  journal = {Journal of Information Processing},
  year = {2024},
  xpages = {1-17},
  note = {to appear},
  xmonth = may,
  organization = {Information Processing Society of Japan},
}

@inproceedings{TSNA25PPDP,
 author    = "Kanta Takahata and Jonas Sch{\"o}pf and Naoki Nishida and Takahito Aoto",
 title     = "Recovering Commutation of Logically Constrained Rewriting and Equivalence Transformations",
 booktitle = {Proceedings of the 27th International Symposium on Principles and
                  Practice of Declarative Programming},
 year      = 2025,
 xpages     = {1-29},
 xdoi       = "",
 publisher    = "ACM",
note     = "to appear",
}

\newpage
\appendix

\section{Detailed Proofs}
\label{sec:missing-proofs}

\subsection{Proofs of \Cref{ssec:existentially cterms and equivalence}}

\LemmaIIIii*
\begin{proof}
We show \Bfnum{1} $\implies$ \Bfnum{2}, \Bfnum{2} $\iff$ \Bfnum{3}, and
\Bfnum{3} $\implies$ \Bfnum{1}.
\begin{itemize}
    \item \Bfnum{1} $\implies$ \Bfnum{2}:
From $\BVar(\ECO{\vec{x}}{\varphi}) \subseteq \Var(\varphi)$
and $\BVar(\ECO{\vec{x}}{\varphi}) \cap \Var(s) = \varnothing$,
we know $\BVar(\ECO{\vec{x}}{\varphi}) \subseteq \Var(\varphi) \setminus \Var(s)$.
On the other hand,
from $\FVar(\ECO{\vec{x}}{\varphi}) \subseteq \Var(\varphi)$ and 
$\FVar(\ECO{\vec{x}}{\varphi}) \subseteq \Var(s)$,
we have $\FVar(\ECO{\vec{x}}{\varphi}) \subseteq \Var(\varphi) \cap \Var(s)$.
Then, as $\FVar(\ECO{\vec{x}}{\varphi})= \Var(\varphi) \setminus \BVar(\ECO{\vec{x}}{\varphi})$
and $\Var(\varphi) \cap \Var(s) = \Var(\varphi) \setminus (\Var(\varphi) \setminus \Var(s))$,
we obtain
$\Var(\varphi) \setminus \BVar(\ECO{\vec{x}}{\varphi})
\subseteq \Var(\varphi) \setminus (\Var(\varphi) \setminus \Var(s))$.
Since $\BVar(\ECO{\vec{x}}{\varphi})$ and $\Var(\varphi) \setminus \Var(s)$
are subsets of $\Var(\varphi)$,
it follows $\Var(\varphi) \setminus \Var(s) \subseteq \BVar(\ECO{\vec{x}}{\varphi})$.
Therefore, $\BVar(\ECO{\vec{x}}{\varphi}) = \Var(\varphi) \setminus \Var(s)$.
%
    \item \Bfnum{2} $\iff$ \Bfnum{3}:
This is clear as $\Var(\varphi) = \BVar(\ECO{\vec{x}}{\varphi}) \uplus \FVar(\ECO{\vec{x}}{\varphi})$
and $\Var(\varphi) = (\Var(\varphi) \setminus \Var(s)) \uplus (\Var(\varphi) \cap \Var(s))$.

    \item \Bfnum{3} $\implies$ \Bfnum{1}:
We have $\FVar(\ECO{\vec{x}}{\varphi}) = \Var(\varphi) \cap \Var(s) \subseteq \Var(s)$.
Also, from the previous item,~\Bfnum{2} also holds.
Thus, from~\Bfnum{2}, we have 
$\BVar(\ECO{\vec{x}}{\varphi}) \cap \Var(s) = 
(\Var(\varphi) \setminus \Var(s)) \cap \Var(s) = \varnothing$.
\qed\end{itemize}
\end{proof}

\LemmaIIIiv*
\begin{proof}
It follows from \Cref{lem:property existential constraints}.
Note that the conditions~\Bfnum{1},\,\Bfnum{2} for existentially constrained terms
equals the condition~\Bfnum{1} of \Cref{lem:property existential constraints}
including the condition on $X$ (i.e., $\FVar(\ECO{\vec{x}}{\varphi}) \subseteq X \subseteq \Var(s)$).
\qed
\end{proof}

\LemmaBasicPropertiesOfXvaluedSubstitution*
\begin{proof}
Since $\CTerm{X}{s}{\vec{x}}{\varphi}$ is satisfiable, 
the constraint $\ECO{\vec{x}}{\varphi}$ is satisfiable.
Thus, there exists a valuation $\rho$ such that 
$\vDash_{\xM,\rho} \ECO{\vec{x}}{\varphi}$.
Take a substitution $\gamma := \rho|_{X}$.
Then, as $\FVar(\ECO{\vec{x}}{\varphi}) \subseteq X$,
we have $\gamma(X \cup \FVar(\ECO{\vec{x}}{\varphi})) \subseteq \Val$,
and $\vDash_\xM (\ECO{\vec{x}}{\varphi})\gamma$.
Thus, $\gamma$ is $X$-valued, $\Dom(\gamma) = X$, and $\gamma \vDash_\xM \ECO{\vec{x}}{\varphi}$.
Moreover, as $\gamma(x) \in  \xV \cup \Val$ for any $x \in \xV$,  we have $\Pos(s) = \Pos(s\gamma)$.
Hence the claim follows.
\qed
\end{proof}

\LemCorrespondence*

\begin{proof}
($\Longrightarrow$)
Suppose $s~\CO{\varphi} \sim t~\CO{\psi}$.
Let 
$X = \Var(s) \cap \Var(\varphi)$,
$\SET{\vec{x}} = \Var(\varphi) \setminus\Var(s)$,
$Y = \Var(t) \cap \Var(\psi)$, and
$\SET{\vec{y}} = \Var(\psi) \setminus\Var(t)$.
Then, obviously
$\CTerm{X}{s}{\vec{x}}{\varphi}$
and $\CTerm{Y}{t}{\vec{y}}{\psi}$ are
existentially constrained terms.
We show
$\CTerm{X}{s}{\vec{x}}{\varphi} \subsetsim \CTerm{Y}{t}{\vec{y}}{\psi}$.
The other direction $\CTerm{X}{s}{\vec{x}}{\varphi} \supsetsim \CTerm{Y}{t}{\vec{y}}{\psi}$
follows similarly.
Let $\sigma$ be an $X$-valued substitution such that
$\sigma \vDash_\xM \ECO{\vec{x}}{\varphi}$.
Then $\sigma(\FVar(\ECO{\vec{x}}{\varphi})) \subseteq \Val$
and $\vDash_\xM \varphi\kappa\sigma$ with $\kappa = \{ \vec{x} \mapsto \vec{v} \}$.
Then $\sigma \circ \kappa \vDash_\xM \varphi$,
and by our assumption there exists a substitution $\gamma$
such that $\gamma \vDash_\xM \psi$ and $s\kappa\sigma = t\gamma$.
Since $\SET{\vec{x}} \cap \Var(s) = \varnothing$ and $\vec{v} \in \Val^*$,
we have $s\sigma = s\kappa\sigma = t\gamma$.
From $\gamma \vDash_\xM \psi$, we know $\gamma(\Var(\psi)) \subseteq \Val$.
As $\FVar(\ECO{\vec{y}}{\psi}) \subseteq \Var(\psi)$, we have $\gamma(\FVar(\psi)) \subseteq \Val$.
Furthermore, it follows from $\vDash_\xM \psi\gamma$ that $\vDash_\xM (\ECO{\vec{y}}{\psi})\gamma$.
This concludes $\gamma \vDash_\xM \ECO{\vec{y}}{\psi}$.

\noindent
($\Longleftarrow$)
Suppose 
$\CTerm{X}{s}{\vec{x}}{\varphi} \sim \CTerm{Y}{t}{\vec{y}}{\psi}$
where 
$X = \Var(s) \cap \Var(\varphi)$,
$\SET{\vec{x}} = \Var(\varphi) \setminus\Var(s)$,
$Y = \Var(t) \cap \Var(\psi)$, and
$\SET{\vec{y}} = \Var(\psi) \setminus\Var(t)$.
We show $s~\CO{\varphi} \subsetsim t~\CO{\psi}$.
As before the other direction $s~\CO{\varphi} \supsetsim t~\CO{\psi}$ follows similarly.
Let $\sigma$ be substitution $\sigma \vDash_\xM \varphi$.
Then, since $\FVar(\ECO{\vec{x}}{\varphi}) \subseteq \Var(\varphi)$
and $\vDash_\xM \varphi\sigma$ implies 
$\vDash_\xM (\ECO{\vec{x}}{\varphi})\sigma$,
we obtain $\sigma \vDash_\xM \ECO{\vec{x}}{\varphi}$.
Also, by $X = \Var(s) \cap \Var(\varphi) \subseteq \Var(\varphi)$, 
it follows that $\sigma$ is $X$-valued.
By our assumption
there exists a $Y$-valued substitution $\gamma$ such that
$\gamma \vDash_\xM \ECO{\vec{y}}{\psi}$ and $s\sigma = t\gamma$.
Thus, 
$\vDash_\xM \psi\kappa\gamma$ with $\kappa = \{ \vec{y} \mapsto \vec{v} \}$
with some $\vec{v} \in \Val^*$.
Then, by $\SET{\vec{y}} \cap \Var(t) = \varnothing$ and $\vec{v} \in \Val^*$,
$s\sigma = t\gamma = t\kappa\gamma = t (\gamma \circ \kappa)$.
Also, we have
$\Var(\psi) = (\Var(\psi) \cap \Var(t)) \cup (\Var(\psi) \setminus\Var(t)) 
= Y \cup \SET{\vec{y}} \subseteq \VDom(\gamma \circ \kappa)$,
i.e.\ $(\gamma \circ \kappa)(\Var(\psi) \subseteq \Val$.
From all of this we obtain $\gamma \circ \kappa \vDash_\xM \psi$.
\qed
\end{proof}

\subsection{Proofs of \Cref{ssec:equivalence characterization of ext cterms}}

\LemmaIIIxxi*
\begin{proof}
As $\delta$ is a bijection, we can take the inverse renaming by $\delta^{-1}$. 
From the condition,
we have $t = s\delta^{-1}$, $Y = \delta^{-1}(X)$,
and $\vDash_\xM (\ECO{\vec{y}}{\psi})\delta^{-1}  \Leftrightarrow (\ECO{\vec{x}}{\varphi})$.
By symmetricity of the definition of equivalence,
it suffices to prove 
$\CTerm{X}{s}{\vec{x}}{\varphi} \subsetsim \CTerm{\SET{ x\delta \mid x \in X }}{s\delta}{}{(\ECO{\vec{x}}{\varphi})\delta}$.
Fix an $X$-valued substitution $\sigma$ with 
$\sigma \vDash_\xM \ECO{\vec{x}}{\varphi}$. 
Let us define the substitution $\gamma := \sigma \circ \delta^{-1}$.
We now show that $\gamma$ is $Y$-valued, 
$\gamma \vDash_\xM \ECO{\vec{y}}{\psi}$ and $s\sigma = t\gamma$. 
Let $y \in Y$. Then, by $Y = \delta(X)$, there exists $x \in X$ such that $\delta(x) = y$.
Hence, as $\sigma$ is $X$-valued,
$y\gamma = x\delta\gamma 
= x\delta\delta^{-1}\sigma = x\sigma \in \Val$.
Thus, we know that $\gamma$ is $Y$-valued.
From $\sigma \vDash_\xM \ECO{\vec{x}}{\varphi}$,
it follows $\vDash_\xM (\ECO{\vec{x}}{\varphi})\sigma$.
Hence, by $\vDash_\xM (\ECO{\vec{x}}{\varphi}) \Leftrightarrow (\ECO{\vec{y}}{\psi})\delta^{-1}$,
we obtain $\vDash_\xM  (\ECO{\vec{y}}{\psi})\delta^{-1}\sigma$.
Therefore,  $\vDash_\xM  (\ECO{\vec{y}}{\psi})\gamma$.
Since $\FVar(\ECO{\vec{y}}{\psi}) \subseteq Y$ and 
$\gamma$ is $Y$-valued,  it follows $\gamma \vDash_\xM \ECO{\vec{y}}{\psi}$.
Finally, 
$t\gamma = (s\delta)\gamma = s\gamma\delta = s\sigma$.
\qed
\end{proof}

\ThmCharacterizationOfEquivalenceRsingRenaming*
\begin{proof}
As the \emph{if} direction follows from \Cref{lem:equivalence by renaming},
we concentrate on the \emph{only-if} direction.
Assume that
$\CTerm{X}{s}{\vec{x}}{\varphi} \sim \CTerm{Y}{s\delta}{\vec{y}}{\psi}$.
We first show by contradiction that $\delta(X) = Y$.
Assume that $\delta(X) \ne Y$, hence $X \neq \delta^{-1}(Y)$.
By $Y \subseteq \Var(t)$, we have $\delta^{-1}(Y) \subseteq 
\delta^{-1}(\Var(t)) = \Var(t\delta^{-1}) = \Var(s)$.
Thus, combining with $X \subseteq \Var(s)$, 
we have that $X\cup \delta^{-1}(Y) \subseteq \Var(s)$.
Thus, there exists a variable $z \in \Var(s)$ 
such that $z \notin X$ and $z \in \delta^{-1}(Y)$
or vice versa.

We first consider the case
that $z \notin X$ and $z \in \delta^{-1}(Y)$.
Let $p$ be a position of $z$ in $s$, i.e., $s|_p=z$.
Since $\CTerm{X}{s}{\vec{x}}{\varphi}$ is satisfiable, 
there exists an $X$-valued substitution $\sigma$ with $\sigma \vDash_\xM \ECO{\vec{x}}{\varphi}$
and $\Dom(\sigma) = X$, by \Cref{lem:basic-properties-of-X-valued-substitutions}.
As $z \notin X = \Dom(\sigma)$, it follows
that $z\sigma \notin \Val$.
By our assumption, there
exist a $Y$-valued substitution $\gamma$ with 
$\gamma \vDash_\xM \ECO{\vec{y}}{\psi}$ and $s\sigma = s\delta\gamma$.
As $\gamma$ is $Y$-valued and $\delta(z) \in Y$,
we have $(s\delta\gamma)|_p =  s|_p\delta\gamma = \delta(z)\gamma \in \Val$.
This contradicts $(s\sigma)|_p = z\sigma \notin \Val$ and $s\sigma = s\delta\gamma$.

As a next step we consider the case
that $z \in X$ and $z \notin \delta^{-1}(Y)$.
Note $\delta(z) \notin \delta(\delta^{-1}(Y)) = Y$ from the latter.
Let $p$ be a position of $z$ in $s$, i.e., $s|_p=z$.
Since $\CTerm{Y}{t}{\vec{y}}{\psi}$ is satisfiable, 
there exists an $Y$-valued substitution $\gamma$ 
with $\gamma \vDash_\xM \ECO{\vec{y}}{\psi}$ and $\Dom(\gamma) = Y$
by \Cref{lem:basic-properties-of-X-valued-substitutions}.
As $\delta(z) \notin Y = \Dom(\gamma)$, it follows
that $z\delta\gamma \notin \Val$.
By our assumption, there
exist a $X$-valued substitution $\sigma$ with 
$\sigma \vDash_\xM \ECO{\vec{x}}{\varphi}$ and $t\gamma = s\sigma$.
As $\sigma$ is $X$-valued and $z \in X$,
we have $(s\sigma)|_p =  s|_p\sigma = z\sigma  \in \Val$.
This contradicts $t\gamma|_p = (s\delta\gamma)|_p = 
s|_p\delta\gamma = z\delta\gamma \notin \Val$ and $s\sigma = t\gamma$.
%
Therefore we have that $\delta(X) = Y$.

In the following we establish
$\vDash_\xM (\ECO{\vec{x}}{\varphi})\delta \Leftrightarrow (\ECO{\vec{y}}{\psi})$.
We first show that
$\vDash_\xM (\ECO{\vec{x}}{\varphi})\delta \Rightarrow (\ECO{\vec{y}}{\psi})$.
Let $\rho$ be a valuation such that
$\vDash_{\xM,\rho} (\ECO{\vec{x}}{\varphi})\delta$.
Then, $\vDash_{\xM} (\ECO{\vec{x}}{\varphi})(\rho \circ \delta)$.
Take a substitution $\sigma := (\rho \circ \delta)|_X$.
Then, clearly, $\sigma$ is $X$-valued.
Moreover, since $\FVar(\ECO{\vec{x}}{\varphi}) \subseteq X$,
we have $\sigma(\FVar(\ECO{\vec{x}}{\varphi})) \subseteq \Val$,
and from $\vDash_{\xM} (\ECO{\vec{x}}{\varphi})\delta\rho$,
it follows $\vDash_{\xM} (\ECO{\vec{x}}{\varphi})\sigma$.
Thus, $\sigma \vDash_{\xM} \ECO{\vec{x}}{\varphi}$.
Using our assumption, one obtains a $Y$-valued substitution 
$\gamma$ with $\gamma \vDash_\xM \ECO{\vec{y}}{\psi}$ 
and $s\sigma = t\gamma = s\delta\gamma$.
Thus, we have that $\sigma|_{\Var(s)} = (\gamma \circ \delta)|_{\Var(s)}$.
That is, we have $\sigma(x) = \gamma(\delta(x))$, for any $x \in \Var(s)$,
Hence, $\sigma(\delta^-1(x)) = \gamma(x)$ for all $x \in \delta(\Var(s)) = \Var(t)$.
Then, as $\FVar(\ECO{\vec{y}}{\psi}) \subseteq \Var(t)$,
it follows from $\gamma \vDash_\xM \ECO{\vec{y}}{\psi}$ 
that $\vDash_\xM (\ECO{\vec{y}}{\psi})\delta^{-1}\sigma$.
By our choice of $\sigma$,
we have $(\rho \circ \delta)|_X = \sigma|_{X}$.
That is, $\rho(\delta(x)) = \sigma(x)$ for all $x \in X$.
Thus, $\rho(y) = \sigma(\delta^{-1}(y))$ for all $y \in \delta(X) = Y$.
Hence, by $\FVar(\ECO{\vec{y}}{\psi}) \subseteq Y$,
it follows $\vDash_\xM (\ECO{\vec{y}}{\psi})\rho$.
Thus, $\vDash_{\xM,\rho} \ECO{\vec{y}}{\psi}$.
Thus, we have shown that 
$\vDash_{\xM,\rho} (\ECO{\vec{x}}{\varphi})\delta$
implies $\vDash_{\xM,\rho} \ECO{\vec{y}}{\psi}$ for any valuation $\rho$,
and therefore
$\vDash_\xM (\ECO{\vec{x}}{\varphi})\delta \Rightarrow (\ECO{\vec{y}}{\psi})$.

To obtain the other direction, i.e., 
$\vDash_\xM (\ECO{\vec{y}}{\psi}) \Rightarrow (\ECO{\vec{x}}{\varphi})\delta$,
we apply the proofs so far to the symmetric proposition 
with replacing $\delta$ by $\delta^{-1}$.
That is, if $t\delta^{-1} = s$ and $\CTerm{Y}{t}{\vec{y}}{\psi} \sim \CTerm{X}{s}{\vec{x}}{\varphi}$
then $\delta^{-1}(Y) = X$ 
and $\vDash_\xM (\ECO{\vec{y}}{\psi})\delta^{-1} \Rightarrow (\ECO{\vec{x}}{\varphi})$.
This works, and 
we obtain $\vDash_\xM (\ECO{\vec{y}}{\psi})\delta^{-1} \Rightarrow (\ECO{\vec{x}}{\varphi})$.
Now it follows from this
that $\vDash_\xM (\ECO{\vec{y}}{\psi}) \Rightarrow (\ECO{\vec{x}}{\varphi})\delta$.
\qed
\end{proof}

\PropEquivalenceByConstraintValidty*
\begin{proof}
Take the identity substitution as $\delta$ in 
\Cref{thm: characterization of equivalence using renaming}.
\qed
\end{proof}

\subsection{Proofs of \Cref{ssec:pattern general constrained terms}}

\LemmaBasicPropertiesOfPG*
\begin{proof}
Let $\Pos_{X\cup\Val}(s) = \SET{p_1,\ldots,p_n}$ and
$\PG(\CTerm{X}{s}{\vec{x}}{\varphi}) = \CTerm{Y}{t}{\vec{y}}{\psi}$ where
\begin{itemize}
    \item $w_1,\ldots,w_n$ are pairwise distinct fresh variables, 
    \item $t = s[w_1,\ldots,w_n]_{p_1,\ldots,p_n}$, 
    \item $Y = \SET{w_1,\ldots,w_n}$,
    \item $\SET{\vec{y}} = \SET{\vec{x}} \cup X$,
        and
    \item $\psi = (\varphi \land \bigwedge_{i=1}^n (s|_{p_i} = w_i))$.
\end{itemize}
We prove the statements in sequence.
\begin{enumerate}
\renewcommand{\labelenumi}{\Bfnum{\arabic{enumi}}}
    \item
To know that $\CTerm{Y}{t}{\vec{y}}{\psi}$ is an existentially constrained term, 
we need to show~(i) $\FVar(\ECO{\vec{y}}{\psi}) \subseteq Y \subseteq \Var(t)$
and~(ii) $\BVar(\ECO{\vec{y}}{\psi}) \cap \Var(t) = \varnothing$.
Firstly, we claim~(i).
We have $\FVar(\ECO{\vec{y}}{\psi}) = \Var(\psi) \setminus \SET{\vec{y}}$.
Then,
$\Var(\psi) 
= \Var(\varphi) \cup \SET{\seq{w}} \cup \SET{ s_{p_i} \in X \mid 1 \le i \le n }
\subseteq \Var(\varphi) \cup \SET{\seq{w}} \cup X$
and
$\Var(\varphi) 
= \FVar(\ECO{\vec{x}}{\varphi}) \cup \SET{\vec{x}}
\subseteq X \cup \SET{\vec{x}} = \SET{\vec{y}}$,
we have 
$\Var(\psi) \setminus \SET{\vec{y}} \subseteq  \SET{\seq{w}} = Y$,
and thus 
$\FVar(\ECO{\vec{y}}{\psi}) \subseteq Y$.
Moreover, as $t = s[w_1,\ldots,w_n]_{p_1,\ldots,p_n}$,
clearly, $Y \subseteq \Var(t)$.
Next, we claim~(ii).
Suppose contrary that $x \in \BVar(\ECO{\vec{y}}{\psi}) \cap \Var(t)$.
Then, since $\SET{\vec{y}} = \SET{\vec{x}} \cup X$,
we have $x \in \SET{\vec{x}} \cap \Var(t)$
or $x \in X \cap \Var(t)$.
Suppose $x \in \SET{\vec{x}} \cap \Var(t)$.
Then, since $t = s[w_1,\ldots,w_n]$ and $\SET{\vec{x}} \cap \SET{\seq{w}} = \varnothing$,
we have $x \in \SET{\vec{x}} \cap \Var(s)$.
But this contradicts that 
$\CTerm{X}{s}{\vec{x}}{\varphi}$ is an existentially constrained term.
Otherwise, $x \in X \cap \Var(t)$.
Again, since $t = s[w_1,\ldots,w_n]$ and $X \cap \SET{\seq{w}} = \varnothing$,
we have $x \in X \cap \Var(s[~]_{\seq{p}})$.
But, since $\Pos_X(s) \subseteq \Pos_{X\cup \Val}(s) = \SET{ \seq{p} }$
we have $X \cap \Var(s[~]_{\seq{p}}) = \varnothing$.
This is a contradiction.
Thus, 
we conclude that $\CTerm{Y}{t}{\vec{y}}{\psi}$ is an existentially constrained term.
    By definition, it is clear that 
    $t$ is $Y$-linear, and 
    $\Val(t) = \varnothing$.
    It follows from \Cref{def:pattern-general}
    that $\CTerm{Y}{t}{\vec{y}}{\psi}$ is pattern-general.

    \item
    It suffices to show that if $\ECO{\vec{x}}{\varphi}$ is satisfiable, then
    $\ECO{\vec{y}}{\psi}$ is so. Assume that $\ECO{\vec{x}}{\varphi}$ is satisfiable.
    Then, there exists a valuation $\rho$ such that $\vDash_{\xM,\rho} \ECO{\vec{x}}{\varphi}$.
    Then there exists $\vec{v} \in \Val^*$ such that 
    $\vDash_{\xM,\rho} \varphi\kappa$ with $\kappa = \SET{\vec{x} \mapsto \vec{v}}$.
    Take $\delta = \SET{ x \mapsto \rho(x) \mid x \in X }$.
    Then, as $X \cap \SET{\vec{x}} = \varnothing$, one can take a substitution $\kappa \cup \delta$.
    Thus, $\vDash_{\xM,\rho} \varphi(\kappa \cup \delta)$.
    Let us define a valuation 
    $\rho' = \rho|_{\xV \setminus Y} \cup \SET{w_i \mapsto \rho(s|_{p_i}) \mid 1 \leqslant i \leqslant n}$.
    Then, clearly, $\vDash_{\xM,\rho'} \bigwedge_{i=1}^{n} (s|_{p_i} = w_i)$.
    Also, by $(X \cup \SET{\vec{x}}) \cap Y = \varnothing$,
    we have $\vDash_{\xM,\rho'} \varphi(\kappa \cup \delta)$.
    Thus, $\vDash_{\xM,\rho'} (\varphi\land \bigwedge_{i=1}^{n} (s|_{p_i} = w_i))(\kappa \cup \delta)$, 
    which implies that $\vDash_{\xM,\rho'} \ECO{\vec{y}}{\psi}$ and therefore 
    $\ECO{\vec{y}}{\psi}$ is satisfiable.
\qed\end{enumerate}
\end{proof}

\TheoremCorrectnessOfPG*
\begin{proof}
Let $\Pos_{X\cup\Val}(s) = \SET{p_1,\ldots,p_n}$ and $\PG(\CTerm{X}{s}{\vec{x}}{\varphi}) = \CTerm{Y}{t}{\vec{y}}{\psi}$
where
\begin{itemize}
    \item $w_1,\ldots,w_n$ are pairwise distinct fresh variables, 
    \item $t = s[w_1,\ldots,w_n]_{p_1,\ldots,p_n}$, 
    \item $Y = \SET{w_1,\ldots,w_n}$,
    \item $\SET{\vec{y}} = \SET{\vec{x}} \cup X$,
        and
    \item $\psi = (\varphi \land \bigwedge_{i=1}^n (s|_{p_i} = w_i))$.
\end{itemize}
We prove $\CTerm{X}{s}{\vec{x}}{\varphi} \subsetsim \CTerm{Y}{t}{\vec{y}}{\psi}$ and $\CTerm{X}{s}{\vec{x}}{\varphi} \supsetsim \CTerm{Y}{t}{\vec{y}}{\psi}$ one after another.

($\subsetsim$)
Let $\sigma$ be an $X$-valued substitution such that 
$\sigma \vDash_{\xM} \ECO{\vec{x}}{\varphi}$. 
Then there exists $\vec{v} \in \Val^*$
such that $\vDash_{\xM} \varphi\kappa\sigma$,
where $\kappa = \SET{\vec{x} \mapsto \vec{v}}$.
Take $\sigma' = \sigma|_{X}$.
Then, since $\sigma(X) \subseteq \Val$,
we have a substitution $\kappa \cup \sigma'\colon \SET{\vec{x}} \cup X \to \Val$
such that $\vDash_{\xM} \varphi(\kappa \cup \sigma')\sigma$.
Take a substitution $\gamma = \sigma \circ \SET{ w_i \mapsto s|_{p_i} \mid 1  \le i \le n }$.
Then, we have $\gamma(w_i) = s|_{p_i}\sigma \in \Val$
as $s|_{p_i}\in X \cup \Val$ and $\sigma$ is $X$-valued.
Hence, $\gamma$ is $Y$-valued.
As $\FVar(\ECO{\vec{y}}{\psi}) \subseteq Y$,
$\gamma(\FVar(\ECO{\vec{y}}{\psi})) \subseteq \Val$ holds.
Also, by our assumption on $w_1,\ldots,w_n$,
it follows from $\vDash_{\xM} \varphi(\kappa \cup \sigma')\sigma$
that $\vDash_{\xM} \varphi(\kappa \cup \sigma')\gamma$.
Also, since 
$\bigwedge_{i=1}^n (s|_{p_i} = w_i))(\kappa \cup \sigma')\gamma
= \bigwedge_{i=1}^n (s|_{p_i}\sigma' = w_i))\gamma
= \bigwedge_{i=1}^n (s|_{p_i}\sigma = w_i\gamma))
= \bigwedge_{i=1}^n (s|_{p_i}\sigma = s|_{p_i}\sigma)$,
we have $\vDash_{\xM} \bigwedge_{i=1}^n (s|_{p_i} = w_i))(\kappa \cup \sigma')\gamma$.
Hence, $\vDash_{\xM} (\varphi \land \bigwedge_{i=1}^n (s|_{p_i} = w_i))(\kappa \cup \sigma')\gamma$,
and thus, $\vDash_{\xM} (\ECO{\vec{y}}{\psi})\gamma$ holds.
Therefore, together with $\gamma(\FVar(\ECO{\vec{y}}{\psi})) \subseteq \Val$,
we obtain $\gamma \vDash_{\xM} \ECO{\vec{y}}{\psi}$.
Finally, as $\Var(s) \cap \SET{\seq{w}} = \varnothing$,
we have $t\gamma = s[w_1,\ldots,w_n]_{\seq{p}}\gamma
= s\gamma[w_1\gamma,\ldots,w_n\gamma]_{\seq{p}}
= s\sigma[s|_{p_1}\sigma,\ldots,s|_{p_n}\sigma]_{\seq{p}}
= s\sigma$.

($\supsetsim$)
Let $\gamma$ be a $Y$-valued substitution such that 
$\gamma \vDash_{\xM} \ECO{\vec{y}}{\psi}$. 
Then, 
there exists $\vec{v} \in \Val^*$
such that $\vDash_{\xM} \psi\kappa\gamma$,
where $\kappa = \SET{\vec{y} \mapsto \vec{v}}$.
Thus, in particular, we have
$\vDash_{\xM} \varphi\kappa\gamma$ and 
$\vDash_{\xM} (s|_{p_i} = w_i)\kappa\gamma$ for all $1 \le i \le n$.
As $\SET{\vec{y}} \cap \SET{\seq{w}} = \varnothing$,
and $s|_{p_i}\in X \cup \Val$ and $X \subseteq \SET{\vec{y}}$,
we have $(s|_{p_i} = w_i)\kappa\gamma = (s|_{p_i}\kappa = w_i\gamma)$.
As $s|_{p_i}\kappa, w_i\gamma \in \Val$,
it follows from $\vDash_{\xM} (s|_{p_i} = w_i)\kappa\gamma$
that $s|_{p_i}\kappa = w_i\gamma$.
Now, take $\kappa' = \kappa|_{\SET{\vec{x}}}$ and $\kappa'' = \kappa|_{X}$.
Then $\kappa = \kappa' \cup \kappa''$.
Also take $\sigma = \gamma \circ \kappa''$.
Firstly, we have $\sigma(s|_{p_i}) 
= s|_{p_i}\kappa''\gamma
= s|_{p_i}\kappa\gamma
= s|_{p_i}\kappa \in \Val$.
Thus, $\sigma$ is $X$-valued, as $X \subseteq \SET{s|_{p_i} \mid 1 \le i \le n}$.
Since $\FVar(\ECO{\vec{x}}{\varphi}) \subseteq X$ and $\sigma$ is $X$-valued, 
$\sigma(\FVar(\ECO{\vec{x}}{\varphi})) \subseteq \Val$.
Moreover, since
$\varphi\kappa\gamma
= \varphi(\kappa' \cup \kappa'')\gamma
= \varphi\kappa'\kappa''\gamma
= \varphi\kappa'\sigma$,
it follows from $\vDash_{\xM} \varphi\kappa\gamma$ 
that $\vDash_{\xM} \varphi\kappa'\sigma$.
Thus, $\vDash_{\xM} (\ECO{\vec{x}}{\varphi})\sigma$.
Hence $\sigma \vDash_\xM \ECO{\vec{x}}{\varphi}$.
Finally, as $\Var(s[~]_{\seq{p}}) \cap X = \varnothing$
and $s|_{p_i}\kappa = w_i\gamma$ ($1 \le i \le n$),
we have 
$s\sigma 
= s[s|_{p_1},\ldots,s|_{p_n}]_{\seq{p}}\sigma
= s[s|_{p_1},\ldots,s|_{p_n}]_{\seq{p}}\kappa''\gamma
= s[s|_{p_1}\kappa,\ldots,s|_{p_n}\kappa]_{\seq{p}}\gamma
= s[w_1\gamma,\ldots,w_n\gamma]_{\seq{p}}\gamma
= s[w_1,\ldots,w_n]_{\seq{p}}\gamma
= t\gamma$.
\qed
\end{proof}

\LemInducedValidityOfRepresentativeSubstitutions*
\begin{proof}
By definition, we have that 
$Y = \SET{ w_1,\ldots,w_n }$,
$\SET{\vec{y}} = \SET{\vec{x}} \cup X$,
and
$\psi = (\varphi \land \bigwedge_{i=1}^n(s|_{p_i}=w_i))$.
Since $w_1,\ldots,w_n$ are freshly introduced variables, we have that $(X\cup\SET{\vec{x}}) \cap \SET{w_1,\ldots,w_n} = \varnothing$. 
W.l.o.g.\ suppose $X = \SET{x_1,\ldots,x_k}$,
$\Pos_X(s) =  \SET{p_1,\ldots,p_k }$ with $s|_{p_i} =  x_i$,
and $\Pos_{\Val}(s) =  \SET{p_{k+1},\ldots,p_n }$ with $s|_{p_i} = v_i$.
Then, for any valuation $\rho$, we have
\[
\begin{array}{rl}
& \vDash_{\xM,\rho} (\ECO{\vec{y}}{\psi})\sigma     \\
\iff& \vDash_{\xM,\rho} (\ECO{\vec{x},x_1,\ldots,x_k}{\varphi \land \bigwedge_{i=1}^n(s|_{p_i}=w_i)})\sigma       \\
\iff& \vDash_{\xM,\rho} (\ECO{\vec{x},x_1,\ldots,x_k}{\varphi \land \bigwedge_{i=1}^k(x_i=w_i) 
\land \bigwedge_{i=k+1}^n(v_i=w_i)})\sigma       \\
\iff& \vDash_{\xM,\rho} (\ECO{\vec{x},x_1',\ldots,x_k'}{\varphi \land \bigwedge_{i=1}^k(x_i'=w_i) 
\land \bigwedge_{i=k+1}^n(v_i=w_i)})\sigma       \\
\iff& \vDash_{\xM,\rho} (\ECO{\vec{x},x_1',\ldots,x_k'}{\varphi \land \bigwedge_{i=1}^k(x_i'=w_i\sigma) 
\land \bigwedge_{i=k+1}^n(v_i=w_i\sigma)})       \\
\iff& \vDash_{\xM,\rho} (\ECO{\vec{x},x_1',\ldots,x_k'}{\varphi \land \bigwedge_{i=1}^k(x_i'=x_i) 
\land \bigwedge_{i=k+1}^n(v_i=v_i)})       \\
\iff& \vDash_{\xM,\rho} (\ECO{\vec{x}}{\varphi})      \\
\end{array}
\]
Hence, $\vDash_\xM (\ECO{\vec{x}}{\varphi}) \Leftrightarrow (\ECO{\vec{y}}{\psi})\sigma$ holds.
\qed
\end{proof}

\LemmaIIIxix*
\begin{proof}
By definition, we have that $Y = \SET{ w_1,\ldots,w_n }$,
$\SET{\vec{y}} = \SET{\vec{x}} \cup X$, and
$\psi = (\varphi \land \bigwedge_{i=1}^n(s|_{p_i}=w_i))$.

\begin{enumerate}
\renewcommand{\labelenumi}{(\arabic{enumi})}

\item
($\Longrightarrow$)
Suppose 
$\vDash_\xM (\ECO{\vec{x}}{\varphi}) \Rightarrow s|_{p_i} = s_{p_j}$.
Let $\rho$ be valuation.
Suppose $\vDash_{\xM,\rho} \psi$, i.e., 
$\vDash_{\xM,\rho} \varphi \land \bigwedge_{i=1}^n(s|_{p_i}=w_i)$.
Then, in particular, we have $\vDash_{\xM,\rho} \varphi$
and $\vDash_{\xM,\rho} (s|_{p_i}=w_i) \land (s|_{p_j}=w_j)$.
Using our assumption, it follows from the former that
$\vDash_{\xM,\rho} (s|_{p_i} = s|_{p_j})$.
Then by the latter, we obtain $\vDash_{\xM,\rho} (w_i = w_j)$.
All in all, we have proved
$\vDash_{\xM,\rho} \psi \Rightarrow (w_i = w_j)$.
From this, 
$\vDash_\xM (\ECO{\vec{y}}{\psi}) \Rightarrow (w_i = w_j)$ follows.

($\Longleftarrow$)
Suppose $\vDash_\xM (\ECO{\vec{y}}{\psi}) \Rightarrow (w_i = w_j)$.
Then $\vDash_\xM \psi \Rightarrow (w_i = w_j)$.
Take a substitution $\sigma = \SET{ w_i \mapsto s|_{p_i} \mid 1 \le i \le n }$.
Then we have $\vDash_\xM \psi\sigma \Rightarrow (w_i\sigma = w_j\sigma)$.
Hence 
$\vDash_\xM \varphi \land \bigwedge_{i=1}^n(s|_{p_i}=s|_{p_i})) \Rightarrow (s|_{p_i}=s|_{p_j})$.
Therefore,
$\vDash_\xM \varphi \Rightarrow (s|_{p_i}=s|_{p_j})$.

\item
($\Longrightarrow$)
Suppose $\vDash_\xM (\ECO{\vec{x}}{\varphi}) \Rightarrow (s|_{p_i} = v)$.
Let $\rho$ be valuation.
Suppose $\vDash_{\xM,\rho} \psi$, i.e.\ 
$\vDash_{\xM,\rho} \varphi \land \bigwedge_{i=1}^n(s|_{p_i}=w_i)$.
Then, in particular, we have $\vDash_{\xM,\rho} \varphi$
and $\vDash_{\xM,\rho} (s|_{p_i}=w_i)$.
Using our assumption, it follows from the former that
$\vDash_{\xM,\rho} (s|_{p_i} = v)$.
Then by the latter, we obtain $\vDash_{\xM,\rho} (w_i = v)$.
All in all, we have proved that $\vDash_\xM \psi \Rightarrow (w_i = v)$.
From this, it follows $\vDash_\xM (\ECO{\vec{y}}{\psi}) \Rightarrow (w_i = v)$.

($\Longleftarrow$)
Suppose $\vDash_\xM (\ECO{\vec{y}}{\psi}) \Rightarrow (w_i = v)$.
Then $\vDash_\xM \psi \Rightarrow (w_i = v)$.
Take a substitution $\sigma = \SET{ w_i \mapsto s|_{p_i} \mid 1 \le i \le n }$.
Then we have $\vDash_\xM \psi\sigma \Rightarrow (w_i\sigma = v)$.
Hence 
$\vDash_\xM \varphi \land \bigwedge_{i=1}^n(s|_{p_i}=s|_{p_i})) \Rightarrow (s|_{p_i}=v)$.
Therefore,
$\vDash_\xM \varphi \Rightarrow (s|_{p_i}=v)$.
\qed
\end{enumerate}
\end{proof}

\LemmaBasicDerivedPropertiesOfEquivalence*
\begin{proof}
\begin{enumerate}
\renewcommand{\labelenumi}{\Bfnum{\arabic{enumi}}}
    \item
    It follows from \Cref{lem:basic-properties-of-X-valued-substitutions} that there exists an $X$-valued substitution $\sigma$ such that 
    $\sigma \vDash_{\xM} \ECO{\vec{x}}{\varphi}$
    and $\Pos(s) = \Pos(s\sigma)$.
    Using the assumption, we obtain
    a $Y$-valued substitution $\gamma$ such that 
    $\gamma \vDash_\xM \ECO{\vec{y}}{\psi}$
    and $s\sigma = t\gamma$.
    Thus, we have that $\Pos(s) = \Pos(s\sigma) =  \Pos(t\gamma)$.
    By definition, it is clear that $\Pos(t) \subseteq \Pos(t\gamma)$ and hence $\Pos(t) \subseteq \Pos(s)$.
    Similarly, we can prove $\Pos(s) \subseteq \Pos(t)$.
    Therefore, $\Pos(s) = \Pos(t)$ holds.
    \item 
    Clearly, it suffices to prove that $p \in \Pos_{X\cup \Val}(s)$ implies $p \in \Pos_{Y\cup \Val}(t)$.
    By \Cref{lem:basic-properties-of-X-valued-substitutions},
    we have $X$-valued substitution $\gamma$ such that $\sigma \vDash_\xM \ECO{\vec{x}}{\varphi}$.
    Then, using the assumption, 
    one obtains a $Y$-valued substitution $\delta$ such that 
    $\delta \vDash_\xM \ECO{\vec{y}}{\psi}$ and $s\gamma = t\delta$.
    Since $s|_p \in X \cup \Val$ and $\gamma$ is $X$-valued,
    we have $s\gamma|_p \in \Val$. Thus, $t\delta|_p = s\gamma|_p \in \Val$.
    Since $\Pos(s) = \Pos(t)$ by~\Bfnum{1}, we know $p \in \Pos(t)$,
    and hence it follows $t|_p \in \xV \cup \Val$.

    To show $t|_p \in Y \cup \Val$, suppose contrary that $t|_p \in \xV \setminus Y$.
    By \Cref{lem:basic-properties-of-X-valued-substitutions},
    we have $Y$-valued substitution $\delta'$ such that $\delta' \vDash_\xM \ECO{\vec{y}}{\psi}$
    and $\Dom(\delta') = Y$.
    Then, since $t|_p \in \Var(t) \setminus Y$, we have $\delta'(t|_p) = t|_p$.
    Using our assumption, one obtains an $X$-valued substitution $\gamma'$
    such that $t\delta' = s\gamma'$.
    Then $s\gamma'|_p = t\delta'|_p = t|_p\delta' = t|_p \in \xV \setminus Y$.
    On the other hand, as $s|_p \in X \cup \Val$ and $\gamma'$ is $X$-valued,
    it follows $s\gamma'|_p \in \Val$. This is a contradiction.

    \item 
    We proceed by contradiction.
    Assume that $\Pos_{\xV\setminus X}(s) \neq \Pos_{\xV\setminus Y}(t)$.
    Then, w.l.o.g., there exists a position $p \in \Pos_{\xV\setminus X}(s)$ 
    such that $p \notin \Pos_{\xV\setminus Y}(t)$.
    Since $\Pos(s) = \Pos(t)$ by~\Bfnum{1}, we know $p \in \Pos(t)$.
    Thus, we have that $p \in \Pos_{\xF \cup Y}(t)$.
    It follows from \Cref{lem:basic-properties-of-X-valued-substitutions} 
    that there exists an $X$-valued substitution $\sigma$ such that 
    $\sigma \vDash_\xM \ECO{\vec{x}}{\varphi}$ and $\Dom(\sigma) = X$.
    Then $(s\sigma)|_p = s|_p \in \xV \setminus X$.
    Using our assumption, one obtains
    a $Y$-valued substitution $\gamma$ such that 
    $\gamma \vDash_\xM \ECO{\vec{y}}{\psi}$ and $s\sigma = t\gamma$.
    Since $p \in \Pos_{\xF\cup Y}(t)$ and $\gamma$ is $Y$-valued, 
    we have that $(t\gamma)|_p \in \xF\cup\Val$. 
    This contradicts the fact that $s\sigma = t\gamma$
    and $(s\sigma)|_p  \in \xV \setminus X$.

    \item 
    This is derived from~\Bfnum{1}--\Bfnum{3} as follows:
    $\Pos_{\xF \setminus \Val}(s)  
    = \Pos(s) \setminus (\Pos_{X \cup \Val}(s) \cup \Pos_{\xV \setminus X}(s))
    = \Pos(t) \setminus (\Pos_{X \cup \Val}(t) \cup \Pos_{\xV \setminus X}(t))
    = \Pos_{\xF \setminus \Val}(t)$.

    \item 
    We proceed with a proof by contradiction. Assume that there exists a
    position $p\in \Pos_{\xF \setminus \Val}(s)$ such that $s(p) \ne t(p)$.
    Then, it follows from~\Bfnum{4} that $p\in \Pos_{\xF \setminus \Val}(t)$.
    By \Cref{lem:basic-properties-of-X-valued-substitutions},
    there exists an $X$-valued substitution $\sigma$ such that 
    $\sigma \vDash_\xM \ECO{\vec{x}}{\varphi}$. 
    Then, using our assumption, one obtains
    a $Y$-valued substitution $\gamma$ such that $\gamma \vDash_\xM \ECO{\vec{y}}{\psi}$ 
    and $s\sigma = t\gamma$.
    Since $p\in \Pos_{\xF \setminus \Val}(s) = \Pos_{\xF \setminus \Val}(t)$, 
    we have that $s(p) = (s\sigma)(p)$ and $t(p) = (t\gamma)(p)$, and hence $s(p) = t(p)$.
    This contradicts the assumption that $s(p) \ne t(p)$.

    \item 
    Let $\sigma$ be an $X$-valued substitution such that 
    $\sigma \vDash_\xM \ECO{\vec{x}}{\varphi}$.
    Then, using our assumption, we obtain a $Y$-valued substitution $\gamma$ 
    such that $\gamma \vDash_\xM \ECO{\vec{y}}{\psi}$ and $s\sigma = t\gamma$.
    Since $t(p) \in \Val$, we have that $(t\gamma)(p) = t(p)$ and hence $\sigma(s(p)) = t(p)$.

    \item 
    It follows from~\Bfnum{1},\,\Bfnum{2} that $\Pos(s[~]_{p_1,\ldots,p_n}) = \Pos(t[~]_{p_1,\ldots,p_n})$.
    For any $p \in \Pos(s[~]_{p_1,\ldots,p_n})$, we have 
    either $p \in \Pos_{\xV \setminus X}(s)$ or $p \in \Pos_{\xF \setminus \Val}(s)$.
    For the former case, $p \in \Pos_{\xV \setminus Y}(t)$ by~\Bfnum{3}.
    For the latter case, $p \in \Pos_{\xF \setminus \Val}(t)$ by~\Bfnum{4},
    and $s|p = t|_p$.
    Thus, $\SET{ p \in \Pos(s[~]_{p_1,\ldots,p_n}) \mid s(p) \neq t(p) } \subseteq
    \Pos_{\xV \setminus X}(s) = \Pos_{\xV \setminus Y}(t)$.
    Thus, it suffices to show for any $q_1,q_2 \in \Pos_{\xV \setminus X}(s)$,
    $s(q_1) = s(q_2)$ if and only if $t(q_1) = t(q_2)$.
    Clearly, it suffices to show that $t(q_1) = t(q_2)$ implies $s(q_1) = s(q_2)$.
    Suppose contrary that $s(q_1) \ne s(q_2)$ and $t(q_1) = t(q_2)$.
    It follows from \Cref{lem:basic-properties-of-X-valued-substitutions} that there exists an $X$-valued substitution $\sigma$ such that 
    $\sigma \vDash_\xM \ECO{\vec{x}}{\varphi}$ and $\Dom(sigma) = X$.
    Since $s(q_1),s(q_2) \notin \Dom(\sigma)$, 
    we have that $(s\sigma)(q_1) = s|_{q_1} \ne s|_{q_2} = (s\sigma)(q_2)$.
    By assumption, there exists a $Y$-valued substitution $\gamma$ such that 
    $\gamma \vDash_\xM \ECO{\vec{y}}{\psi}$ and $s\sigma = t\gamma$.
    Since $t(q_1) = t(q_2)$, we have that $(t\gamma)(q_1) = (t\gamma)(q_2)$.
    This contradicts the fact that $s\sigma = t\gamma$ and $(s\sigma)(q_1) \ne (s\sigma)(q_2)$.
\qed\end{enumerate}
\end{proof}

\LemChacterizationOfPatternGeneral*
\begin{proof}
We prove the \textit{if} part by contradiction. Assume that $s$ is most general
in $\SET{ t \mid \CTerm{X}{s}{\vec{x}}{\varphi} \sim \CTerm{Y}{t}{\vec{y}}{\psi}
}$ and $\CTerm{X}{s}{\vec{x}}{\varphi}$ is not pattern-general (i.e., $s$ is
not linear w.r.t.\ $X$ or a value $c$ appears in $s$). We proceed by case analysis
on whether $s$ is linear w.r.t.\ $X$:
\begin{itemize}
    \item Assume $s$ is linear w.r.t.\ $X$.
    In this case, a value $c$ appears in $s$.
    Let $s = s[c]_p$ for some position $p$ in $s$.
    Let $s' = s[y]_p$ for some fresh variable $y$.
    Then, $\CTerm{X}{s'}{\vec{x}}{\varphi\land(y = c)}$ is an existentially constrained term.
    By definition, it is clear that $\CTerm{X}{s}{\vec{x}}{\varphi} \sim \CTerm{X}{s'}{\vec{x}}{\varphi\land(y = c)}$.
    It is also clear that $s'$ is more general than $s$ but $s$ is not more general than $s'$.
    This contradicts the assumption that $s$ is most general in $\SET{t \mid \CTerm{X}{s}{\vec{x}}{\varphi} \sim \CTerm{Y}{t}{\vec{y}}{\psi}}$.
    \item Assume $s$ is not $X$-linear.
    In this case, there exists a variable $y$ with non-linear occurrences in $s$.
    Let $p_1,\ldots,p_n$ be the pairwise distinct positions in $s$ such that
    $s|_{p_1} = \cdots = s|_{p_n} = y$.
    Let $z_2,\ldots,z_n$ be fresh pairwise distinct variables.
    Let $s' = s[y,z_2,\ldots,z_n]_{p_1,\ldots,p_n}$.
    Then, $\CTerm{X\cup \SET{z_2,\ldots,z_n}}{s'}{\vec{x}}{(\varphi\land \bigwedge_{i=2}^n(z_i = y))}$ is an existentially constrained term.
    By definition, it is clear that $\CTerm{X}{s}{\vec{x}}{\varphi} \sim \CTerm{X}{s'}{\vec{x}}{(\varphi\land\bigwedge_{i=2}^n(z_i = y))}$.
    It is also clear that $s'$ is more general than $s$ but $s$ is not more general than $s'$.
    This contradicts the assumption that $s$ is most general in $\SET{ t \mid \CTerm{X}{s}{\vec{x}}{\varphi} \sim \CTerm{Y}{t}{\vec{y}}{\psi} }$.
\end{itemize}

For the \textit{only-if} part, 
assume that $\CTerm{X}{s}{\vec{x}}{\varphi}$ is satisfiable and pattern-general,
and $\CTerm{X}{s}{\vec{x}}{\varphi} \sim \CTerm{Y}{t}{\vec{y}}{\psi}$.
Then, it follows from \Cref{lem:basic-derived-properties-of-equivalence}~(2) 
that $\Pos_{\Val \cup X} (s) = \Pos_{\Val \cup Y} (t) = \SET{\seq{p}}$
for some $\seq{p}$. 
It also follows from \Cref{lem:basic-derived-properties-of-equivalence}~(7) 
that
$\rho(s[~]_{\seq{p}}) =  t[~]_{\seq{p}}$ for some renaming
$\rho: (\Var(s)\setminus X) \to (\Var(t)\setminus Y)$.
Let $s = s[\seq{s}]_{\seq{p}}$.
Then, because $\CTerm{X}{s}{\vec{x}}{\varphi}$ is pattern-general,
$\seq{s}$ consists of mutually distinct variables.
Thus, one can take a substitution $\delta: \SET{\seq{s}} \to \SET{ t|_{p_i} \mid 1 \le i \le n }$
such that $\delta(s_i) = t|_{p_i}$ ($1 \le i \le n$).
Moreover, since  $\seq{s} \in X$,
$\SET{\seq{s}} \cap (\Var(s)\setminus X) = \varnothing$.
Thus, we can take a substitution $\rho \cup \gamma$.
Then, $s(\rho \cup \gamma) 
= s[s_1,\ldots,s_n](\rho \cup \gamma) 
= s\rho[s_1\gamma,\ldots,s_n\gamma]
= t[t_1,\ldots,t_n] = t$.
Therefore, we conclude $\CTerm{X}{s}{\vec{x}}{\varphi}$ is most general
in the set $\SET{t \mid \CTerm{X}{s}{\vec{x}}{\varphi} \sim \CTerm{Y}{t}{\vec{y}}{\psi}}$.
\qed
\end{proof}

\TheoremIIIxxviii*
\begin{proof}
The \textit{if} part follows from 
\Cref{lem:equivalence by renaming}.
We proceed to prove the \textit{only-if} part.
Suppose $\CTerm{X}{s}{\vec{x}}{\varphi}$ and $\CTerm{Y}{t}{\vec{y}}{\psi}$ are pattern-general
and $\CTerm{X}{s}{\vec{x}}{\varphi} \sim \CTerm{Y}{t}{\vec{y}}{\psi}$.
Then, no value appears in $s$ or $t$ and hence $\Pos_{\Val}(s) = \Pos_{\Val}(t) = \varnothing$.
Thus, by \Cref{lem:basic-derived-properties-of-equivalence} all of the following statements hold:
\begin{itemize}
    \item $\Pos(s) = \Pos(t)$,
    \item $\Pos_X(s) = \Pos_Y(t)$,
    \item $\Pos_{\xV\setminus X}(s) = \Pos_{\xV\setminus Y}(t)$,
    \item $\Pos_{\xF}(s) = \Pos_{\xF\setminus \Val}(s) = \Pos_{\xF \setminus \Val}(t) = \Pos_{\xF}(t)$,
    \item $s|_p = t|_p$ for any position $p\in \Pos_{\xF \setminus \Val}(s)$ ($= \Pos_{\xF}(s) = \Pos_{\xF}(t)$),
        and
    \item there exists a renaming $\theta\colon \Var(s) \setminus X \to \Var(t) \setminus Y$ such that 
    $\theta(s[\,]_{p_1,\ldots,p_n}) = t[\,]_{p_1,\ldots,p_n}$.
\end{itemize}
Let $\SET{p_1,\ldots,p_n} = \Pos_{X}(s) =\Pos_{Y}(t)$,
and $s|_{p_i} = z_i$ and $t|_{p_i} = w_i$ for $1 \leqslant i \leqslant n$.
As $\CTerm{X}{s}{\vec{x}}{\varphi}$ and $\CTerm{Y}{t}{\vec{y}}{\psi}$ are pattern-general,
$z_1,\ldots,z_n$ are pairwise distinct variables, and so are $w_1,\ldots,w_n$.
Let $\rho = \theta \cup \SET{ z_i \mapsto w_i \mid 1 \leq i \leq n}$.
Since $\theta$ is a renaming,
$\SET{z_1,\ldots,z_n} \cap (\Var(s) \setminus X) = \varnothing$, and
$\SET{w_1,\ldots,w_n} \cap (\Var(t) \setminus Y) = \varnothing$,
we know $\rho$ is well-defined and renaming.
Furthermore, we have $s\rho = (s[z_1,\ldots,z_n]_{p_1,\ldots,p_n})\rho 
= s\rho[z_1\rho,\ldots,z_n\rho]_{p_1,\ldots,p_n} 
= s\theta[w_1,\ldots,w_n]_{p_1,\ldots,p_n} 
= t[w_1,\ldots,w_n]_{p_1,\ldots,p_n} = t$.
The rest of the statements follow 
from this by \Cref{thm: characterization of equivalence using renaming}.
\qed
\end{proof}

\subsection{Proofs of \Cref{sec:general-characterization-of-equivalence}}

\LemPropertiesOfEquivalenceOverXvarValPositions*
\begin{proof}
\begin{enumerate}
\renewcommand{\labelenumi}{\Bfnum{\arabic{enumi}}}
    \item 
    By definition, it is clear that $\sim_{\Pos_{X\cup \Val}(s)}$ is reflexive and symmetric.
    We show the transitivity of $\sim_{\Pos_{X\cup \Val}(s)}$.
    Assume that $p_1 \sim_{\Pos_{X\cup \Val}(s)} p_2$ and $p_2 \sim_{\Pos_{X\cup \Val}(s)} p_3$.
    Then, by definition, we have that 
    $\vDash_\xM ((\ECO{\vec{x}}{\varphi}) \Rightarrow s|_{p_1} = s|_{p_2})$
    and
    $\vDash_\xM ((\ECO{\vec{x}}{\varphi}) \Rightarrow s|_{p_2} = s|_{p_3})$ and hence
    $\vDash_\xM ((\ECO{\vec{x}}{\varphi}) \Rightarrow s|_{p_1} = s|_{p_3})$.
    Therefore, we have that $p_1 \sim_{\Pos_{X\cup \Val}(s)} p_3$ and thus $\sim_{\Pos_{X\cup \Val}(s)}$ is transitive.
    
    \item 
    Let $p,q \in \Pos_{X\cup \Val}(s)$ such that $s|_p = s|_q$.
    Then, $\vDash_\xM (s|_p = s|_q)$ and hence 
    $\vDash_\xM ((\ECO{\vec{x}}{\varphi}) \Rightarrow s|_p = s|_q)$.
    Therefore, we have that $p \sim_{\Pos_{X\cup \Val}(s)} q$.
\qed\end{enumerate}
\end{proof}

\LemPropertiesOfPosExcl*
\begin{proof}
\begin{enumerate}
\renewcommand{\labelenumi}{\Bfnum{\arabic{enumi}}}
    \item 
    By definition, it is clear that 
    there exists at least one value $v$ such that 
    $\vDash_\xM ((\ECO{\vec{x}}{\varphi}) \Rightarrow s|_p = v)$.
    Assume 
    $\vDash_\xM ((\ECO{\vec{x}}{\varphi}) \Rightarrow s|_p = v')$.  
    Since $\vDash_\xM \ECO{\vec{x}}{\varphi}$ is satisfiable, 
    $\vDash_{\xM,\rho} \ECO{\vec{x}}{\varphi}$ for some valuation $\rho$.
    Then, we have $\vDash_{\rho,\xM}  s|_p = v$
    and $\vDash_{\rho,\xM}  s|_p = v'$.
    Hence, $v = \rho(s|_p) = v'$.

    \item 
    Let $p,q \in \Pos_{X \cup \Val}(s)$ such that $p \in \Pos_{\Val!}(s)$ and $p \sim q$.
    Then, from $p \in \Pos_{\Val!}(s)$, 
    there exists a value $v \in \Val$ such that 
    $\vDash_\xM ((\ECO{\vec{x}}{\varphi}) \Rightarrow s|_{p} = v)$, 
    and
    from $p \sim q$, 
    we have $\vDash_\xM ((\ECO{\vec{x}}{\varphi}) \Rightarrow s|_p = s|_q)$.
    Thus, we have that $\vDash_\xM ((\ECO{\vec{x}}{\varphi}) \Rightarrow s|_q = v)$.
    Therefore, we obtain $q \in \Pos_{\Val!}(s)$.

    \item 
    Let $p,q \in \Pos_{X\cup\Val}(s)$ and $v,v' \in \Val$ such that
    $p \sim q$, $\vDash_\xM ((\ECO{\vec{x}}{\varphi}) \Rightarrow s|_p = v)$, 
    and $\vDash_\xM ((\ECO{\vec{x}}{\varphi}) \Rightarrow s|_q = v')$.
    Then, 
    from $p \sim q$
    we have that $\vDash_\xM ((\ECO{\vec{x}}{\varphi}) \Rightarrow s|_p = s|_q)$. 
    Hence,  
    $\xM \vDash ((\ECO{\vec{x}}{\varphi}) \Rightarrow v = v')$.
    Therefore, we obtain $v = v'$.
\qed\end{enumerate}
\end{proof}

\LemInducedEqualitiesOfExpliciitSubstitutions*
\begin{proof}
Let $z \in X$.
Then, Since $X \subseteq \Var(s)$, there exists a position $p \in \Pos_X(s)$ such that $s|_p = z$.
We make a case analysis depending on the existence of a position $p \in \Pos_{\Val!}(s)$ such that $s|_p = z$.
\begin{itemize}
    \item Assume that there exists a position $p \in \Pos_{\Val!}(s)$ such that $s|_p=z$.
    Then, by definition, we have that $\mu_{X}(z) = \Val!(p)$.
    Since $p \in \Pos_{\Val!}(s)$ and $s|_p=z$, by definition, 
    we have that $\vDash_\xM (\ECO{\vec{x}}{\varphi}) \Rightarrow (z=\Val!(p))$
    and hence $\vDash_\xM (\ECO{\vec{x}}{\varphi}) \Rightarrow (z = \mu_{X}(z))$.
    \item Assume that there exists no position $p \in \Pos_{\Val!}(s)$ such that $s|_p=z$.
    Then, by definition, we have that $\mu_{X}((z) = s|_{\hat{p}}$.
    Since $\hat{p}$ is a representative of $[p]_\sim$, we have that $p \sim \hat{p}$.
    It follows from \Cref{lem:properties-of-PosValExcl} that $\vDash_\xM (\ECO{\vec{x}}{\varphi}) \Rightarrow (s|_p = s|_{\hat{p}})$
    and hence
     $\vDash_\xM (\ECO{\vec{x}}{\varphi}) \Rightarrow (z = \mu_{X}(z))$.
\qed\end{itemize}
\end{proof}

\LemInducedPropertisOfEuivalence*
\begin{proof}
Let $\PG(\CTerm{X}{s}{\vec{x}}{\varphi}) = (\CTerm{\SET{w_1,\ldots,w_n}}{s[w_1,\ldots,w_n]_{p_1,\ldots,p_n}}{\pvec{x}'}{\varphi'})$
and
$\PG(\CTerm{Y}{t}{\vec{y}}{\psi}) = (\CTerm{\SET{w'_1,\ldots,w'_n}}{t[w'_1,\ldots,w'_n]_{p_1,\ldots,p_n}}{\pvec{y}'}{\psi'})$, 
where
$\SET{\pvec{x}'} = \SET{\vec{x}}\cup X$,
$\varphi' = (\varphi \land \bigwedge_{i=1}^n (w_i = s|_{p_i}))$,
$\SET{\pvec{y}'} = \SET{\vec{y}}\cup Y$,
and
$\psi' = (\psi \land \bigwedge_{i=1}^n (w'_i = t|_{p_i}))$.
Then, by our assumption and \Cref{thm:correctness-of-PG}, we have
\[
\begin{array}{cl}
&(\CTerm{\SET{w_1,\ldots,w_n}}{s[w_1,\ldots,w_n]_{p_1,\ldots,p_n}}{\pvec{x}'}{\varphi'})\\
\sim&
(\CTerm{\SET{w'_1,\ldots,w'_n}}{t[w'_1,\ldots,w'_n]_{p_1,\ldots,p_n}}{\pvec{y}'}{\psi'}). 
\end{array}
\]
Hence, by \Cref{thm:complete characterization of equivalence for most sim-general constrained terms}, 
there is a renaming $\delta$ such that 
$s[w_1,\ldots,w_n]\delta = t[w_1',\ldots,w_n']$
and $\vDash_\xM (\ECO{\pvec{x}'}{\varphi'})\delta \Leftrightarrow \ECO{\pvec{y}'}{\psi'}$.
Thus, $\delta(w_i) = w_i'$. 
Then, using \Cref{lem:X or val subterms equality from constraints}~\Bfnum{1},
\[
\begin{array}{lcl}
 \vDash_\xM (\ECO{\vec{x}}{\varphi}) \Rightarrow (s|_{p_i} = s|_{p_j})
&\iff& \vDash_\xM (\ECO{\pvec{x}'}{\varphi'}) \Rightarrow (w_i = w_j)\\
&\iff& \vDash_\xM (\ECO{\pvec{x}'}{\varphi'})\delta \Rightarrow (w_i\delta = w_j\delta) \\
&\iff& \vDash_\xM (\ECO{\pvec{y}'}{\psi'}) \Rightarrow (w_i' = w_j')\\
&\iff& \vDash_\xM (\ECO{\pvec{y}}{\psi}) \Rightarrow (t|_{p_i} = t|_{p_j})\\
\end{array}
\]
This proves~\Bfnum{1}. \Bfnum{2} follows similarly, by
using \Cref{lem:X or val subterms equality from constraints}~\Bfnum{2}.
In order to show~\Bfnum{3}, take $\theta = \SET{ \langle s|_{p_i}, t|_{p_i} \rangle \mid 1 \le i \le n }$.
Then, by~\Bfnum{1},\,\Bfnum{2}, and our assumption, we have
\begin{itemize}
\item $\mu_X(s|_{p_i}) \in \Val$ if and only if  $\mu_Y(t|_{p_i}) \in \Val$, and moreover, 
\item $s|_{p_i} \in \tilde{X}$ if and only if $p_i$ is a representative of an equivalence class $[p_i]$
if and only if $t|_{p_i} \in \tilde{Y}$.
\end{itemize}
Thus, 
$\{ \langle s|_{p_i}, t|_{p_i} \rangle \mid 1 \le i \le n$, 
$p_i$ is a representative of an equivalence class $[p_i]$
such that $\mu_X(s|_{p_i}) \notin \Val \} = \theta|_{\tilde{X}}$
is a renaming from $\tilde{X}$ to $\tilde{Y}$.
Take an arbitrary but fixed function $\theta': X \to Y$
such that $\theta'(s|_{p_i}) \in \SET{ t|_{p_j} \mid p_i \sim p_j }$.
Then (i) $\theta|_{\tilde{X}}(\mu_X(s|_{p_i})) = \theta|_{\tilde{X}}(s|_{\hat{p_i}}))) 
= t|_{\hat{p_i}} = \mu_Y(t|_{p_j}) = \mu_Y(\theta'(s|_{p_i}))$.
Let $\sigma = \SET{ w_1 \mapsto s|_{p_1},\ldots, w_n \mapsto s|_{p_n}}$
and
$\sigma' = \SET{ w'_1 \mapsto t|_{p_1},\ldots, w'_n \mapsto t|_{p_n}}$.
Then, we have
(ii) $\mu_Y(\theta'(\sigma(w_i)))
= \mu_Y(\theta'(s|_{p_i}))
= \mu_Y(t|_{p_j})
= t|_{\hat{p_i}}
= \mu_Y(t|_{p_i})
= \mu_Y(\sigma'(w_i'))
= \mu_Y(\sigma'(\delta(w_i)))$.
Thus, for any valuation $\rho$,
\[
\begin{array}{lcl@{\qquad}l}
 \vDash_{\xM,\rho} (\ECO{\vec{x}}{\varphi})\mu_X\theta|_{\tilde{X}}
&\iff& \vDash_{\xM,\rho} (\ECO{\pvec{x}'}{\varphi'})\sigma\mu_X\theta|_{\tilde{X}}
& \mbox{by \Cref{lem:induced-validity-of-representative-substitutions}}\\
&\iff& \vDash_{\xM,\rho} (\ECO{\pvec{x}'}{\varphi'})\sigma\theta'\mu_Y
& \mbox{by (i)}\\
&\iff& \vDash_{\xM,\rho}(\ECO{\pvec{x}'}{\varphi'})\delta\sigma'\mu_Y
& \mbox{by (ii)}\\
&\iff& \vDash_{\xM,\rho} (\ECO{\pvec{y}'}{\psi'})\sigma'\mu_Y
& \mbox{by \Cref{lem:induced-validity-of-representative-substitutions}}\\
&\iff& \vDash_{\xM,\rho} (\ECO{\pvec{y}}{\psi}) \mu_Y
\end{array}
\]
Thus, 
$\vDash_{\xM} (\ECO{\vec{x}}{\varphi})\mu_X\theta|_{\tilde{X}} \Leftrightarrow (\ECO{\pvec{y}}{\psi}) \mu_Y$.
\qed
\end{proof}

\ThmCompleteCharacterizatinOfEquivalenceOfConstrainedTerms*
\begin{proof}
The \textit{only-if} part follows \Cref{lem:basic-derived-properties-of-equivalence} 
and \Cref{lem:induced-propertis-of-equivalence}.
We show the \textit{if} part.
Let $\PG(\CTerm{X}{s}{\vec{x}}{\varphi}) = (\CTerm{X'}{s'}{\pvec{x}'}{\varphi'})$
and
$\PG(\CTerm{Y}{t}{\vec{y}}{\psi}) = (\CTerm{Y'}{t'}{\pvec{y}'}{\psi'})$,
where
$X' = \SET{ w_1,\ldots,w_n }$,
$s' = s[w_1,\ldots,w_n]_{\seq{p}}$, 
$\SET{\pvec{x}'} = \SET{\vec{x}}\cup X$,
$\varphi' = (\varphi \land \bigwedge_{i=1}^n w_i = s|_{p_i})$,
$Y' = \SET{ w_1',\ldots,w_n' }$,
$t' = s[w_1',\ldots,w_n']_{\seq{p}}$, 
$\SET{\pvec{y}'} = \SET{\vec{y}} \cup Y$, and
$\psi' = (\psi \land \bigwedge_{i=1}^n w'_i = t|_{p_i})$.
Then, by \Cref{thm:correctness-of-PG}, we have
$\CTerm{X}{s}{\vec{x}}{\varphi} \sim \CTerm{X'}{s'}{\pvec{x}'}{\varphi'}$, and
$\CTerm{Y}{t}{\vec{y}}{\psi} \sim \CTerm{Y'}{t'}{\pvec{y}'}{\psi'}$.
Thus, it suffices to show that
$\CTerm{X'}{s'}{\pvec{x}'}{\varphi'} \sim \CTerm{Y'}{t'}{\pvec{y}'}{\psi'}$.
Since both of these constrained terms are pattern-general,
we show this claim using 
\Cref{thm:complete characterization of equivalence for most sim-general constrained terms}.

By our assumptions~\Bfnum{1},\,\Bfnum{3},
and applying \Cref{lem:induced-propertis-of-equivalence}
to $\CTerm{X}{s}{\vec{x}}{\varphi} \sim \CTerm{X'}{s'}{\pvec{x}'}{\varphi'}$
and $\CTerm{Y}{t}{\vec{y}}{\psi} \sim \CTerm{Y'}{t'}{\pvec{y}'}{\psi'}$,
we know
${\sim}_{X \cup \Var(s)} = {\sim}_{X' \cup \Var(s')} = 
{\sim}_{Y' \cup \Var(t')} = {\sim}_{Y \cup \Var(t)}$.
Similarly, by~\Bfnum{4} and \Cref{lem:induced-propertis-of-equivalence}, 
one can take the representative substitutions $\mu_{X'},\mu_{Y'}$ of
$\CTerm{X'}{s'}{\pvec{x}'}{\varphi'}$ and $\CTerm{Y'}{t'}{\pvec{y'}}{\psi'}$, respectively, 
based on the same representative for
each equivalence class $[p_i]_\sim$ ($1 \le i \le n$).

For using \Cref{thm:complete characterization of equivalence for most sim-general constrained terms},
take 
$\xi = \SET{ w_i \mapsto w_i' \mid 1 \le i \le n }$,
and $\nu = \rho \cup \xi$.
Then, clearly, we have
$s'\nu = s[w_1,\ldots,w_n]\nu = s\rho[w_1',\ldots,w_n'] = t[w_1',\ldots,w_n'] = t'$,
and $\nu(X') = Y'$.
It remains to show 
$\vDash_\xM (\ECO{\pvec{x}'}{\varphi'})\nu \Leftrightarrow (\ECO{\pvec{y}'}{\psi'})$.

From our assumption, 
we know
$\vDash_{\xM} (\ECO{\vec{x}}{\varphi})\mu_X\theta|_{\tilde{X}}
\Leftrightarrow (\ECO{\vec{y}}{\varphi})\mu_Y$,
where $\theta = \SET{ \langle s|_{p_i}, t|_{p_i} \rangle \mid 1 \le i \le n }$.
Again, applying \Cref{lem:induced-propertis-of-equivalence}
to $\CTerm{X}{s}{\vec{x}}{\varphi} \sim \CTerm{X'}{s'}{\pvec{x}'}{\varphi'}$
and $\CTerm{Y}{t}{\vec{y}}{\psi} \sim \CTerm{Y'}{t'}{\pvec{y}'}{\psi'}$,
we also have 
$\vDash_{\xM} (\ECO{\vec{x}}{\varphi})\mu_X\sigma|_{\tilde{X}}
\Leftrightarrow (\ECO{\pvec{x}'}{\varphi'})\mu_{X'}$
and
$\vDash_{\xM} (\ECO{\vec{y}}{\psi})\mu_Y\sigma'|_{\tilde{Y}}
\Leftrightarrow (\ECO{\pvec{y}'}{\psi'})\mu_{Y'}$,
where $\sigma = \SET{ \langle s|_{p_i}, w_i \rangle \mid 1 \le i \le n }$
and $\sigma' = \SET{  \langle t|_{p_i}, w_i' \rangle  \mid 1 \le i \le n }$.
From these, it easily follows that 
$\vDash_{\xM,\gamma} (\ECO{\pvec{x}'}{\varphi'})\mu_{X'}\nu
\Leftrightarrow (\ECO{\pvec{y}'}{\varphi'})\mu_{Y'}$.
Also, by \Cref{lem:induced-equalities-of-expliciit-substitutions}, we have
$\vDash_{\xM} \varphi \Rightarrow (\bigwedge_{i=1}^n (s|_{p_i} = s|_{\hat{p_i}}))$
and 
$\vDash_{\xM} \psi \Rightarrow (\bigwedge_{i=1}^n (t|_{p_i} = t|_{\hat{p_i}}))$,
as $\mu_X(s|_{p_i}) = s|_{\hat{p_i}}$ and $\mu_Y(t|_{p_i}) = t|_{\hat{p_i}}$.
Thus, for any valuation $\gamma$,
\[
\begin{array}{cl}
& \vDash_{\xM,\gamma} (\ECO{\pvec{x}'}{\varphi'})\nu\\
\iff&
\vDash_{\xM,\gamma} (\ECO{\pvec{x}'}{\varphi 
\land (\bigwedge_{i=1}^n (w_i = s|_{p_i}))})\nu\\
\iff&
\vDash_{\xM,\gamma} (\ECO{\pvec{x}'}{\varphi
\land (\bigwedge_{i=1}^n (w_i = s|_{p_i}))
\land (\bigwedge_{i=1}^n (s|_{p_i} = s|_{\hat{p_i}}))})\nu\\
\iff&
\vDash_{\xM,\gamma} (\ECO{\pvec{x}'}{\varphi
\land (\bigwedge_{i=1}^n (\hat{w_i} = s|_{p_i}))})\nu\\
= &
\vDash_{\xM,\gamma} (\ECO{\pvec{x}'}{\varphi'})\mu_{X'}\nu\\
\iff&
\vDash_{\xM,\gamma} (\ECO{\pvec{y}'}{\psi'})\mu_{Y'}\\
\iff&
\vDash_{\xM,\gamma} \ECO{\pvec{y}'}{\psi
\land (\bigwedge_{i=1}^n (\hat{w_i'} = t|_{p_i}))}\\
\iff&
\vDash_{\xM,\gamma} \ECO{\pvec{y}'}{\psi
\land (\bigwedge_{i=1}^n (\hat{w_i'} = t|_{p_i}))
\land (\bigwedge_{i=1}^n (t|_{p_i} = t|_{\hat{p_i}}))}\\
\iff&
\vDash_{\xM,\gamma} \ECO{\pvec{y}'}{\psi
\land (\bigwedge_{i=1}^n (w_i' = t|_{p_i}))}\\
=&
\vDash_{\xM,\gamma} \ECO{\pvec{y}'}{\varphi'}
\end{array}
\]
Hence $\vDash_\xM (\ECO{\pvec{x}'}{\varphi'})\nu \Leftrightarrow (\ECO{\pvec{y}'}{\psi'})$.
Thus, 
from \Cref{thm:complete characterization of equivalence for most sim-general constrained terms}
we conclude $\CTerm{X'}{s'}{\pvec{x}'}{\varphi'} \sim \CTerm{Y'}{t'}{\pvec{y}'}{\psi'}$.
Hence,
$\CTerm{X}{s}{\vec{x}}{\varphi} \sim \CTerm{Y}{t}{\vec{y}}{\psi}$ follows.
\qed
\end{proof}

\end{document}